\documentclass[journal]{IEEEtran}
\usepackage{graphicx}
\usepackage{bm}
\usepackage{amsfonts}
\usepackage{amsthm}
\usepackage{amssymb}
\usepackage{MyMnSymbol}
\usepackage{color}
\usepackage{mathrsfs}
\usepackage{flushend}
\usepackage{cite}
\usepackage{booktabs}
\usepackage{amsmath}
\usepackage{amsthm}
\usepackage{algorithm}
\usepackage{algorithmic}
\ifCLASSOPTIONcompsoc
 \usepackage[caption=false,font=normalsize,labelfont=sf,textfont=sf]{subfig}
\else
 \usepackage[caption=false,font=footnotesize]{subfig}
\fi
\newtheorem{thm}{Theorem}
\newtheorem{lemma}{Lemma}
\newtheorem{corollary}{Corollary}

\newtheorem{proposition}{Proposition}

\ifCLASSINFOpdf
\else
\fi
\begin{document}
%
\title{Cram\'er-Rao Bound Optimization for Joint Radar-Communication Design}
%
%
%

\author{Fan Liu,~\IEEEmembership{Member,~IEEE,}
        Ya-Feng Liu,~\IEEEmembership{Senior Member,~IEEE,}
        Ang Li,~\IEEEmembership{Member,~IEEE,}\\
        Christos Masouros,~\IEEEmembership{Senior~Member,~IEEE,}
        and~Yonina C. Eldar,~\IEEEmembership{Fellow,~IEEE}
\thanks{F. Liu is with the Department of Electrical and Electronic Engineering, Southern University of Science and Technology, Shenzhen 518055, China (e-mail: liuf6@sustech.edu.cn).}
\thanks{Y.-F. Liu is with the State Key Laboratory of Scientific and Engineering Computing, Institute of Computational Mathematicsand Scientific/Engineering Computing, Academy of Mathematics and Systems Science, Chinese Academy of Sciences, Beijing 100190, China (e-mail: yafliu@lsec.cc.ac.cn).}
\thanks{A. Li is with the School of Information and Communications Engineering, Faculty of Electronic and Information Engineering, Xi'an Jiaotong University, Xi'an, Shaanxi 710049, China. (e-mail: ang.li.2020@xjtu.edu.cn)}
\thanks{C. Masouros is with the Department of Electronic and Electrical Engineering, University College London, London, WC1E 7JE, UK (e-mail: chris.masouros@ieee.org).}
\thanks{Y. C. Eldar is with the Faculty of Mathematics and Computer Science, Weizmann Institute of Science, Rehovot, Israel (e-mail: yonina.eldar@weizmann.ac.il).}
}

\maketitle

\begin{abstract}
In this paper, we propose multi-input multi-output (MIMO) beamforming designs towards joint radar sensing and multi-user communications. We employ the Cram\'er-Rao bound (CRB) as a performance metric of target estimation, under both point and extended target scenarios. We then propose minimizing the CRB of radar sensing while guaranteeing a pre-defined level of signal-to-interference-plus-noise ratio (SINR) for each communication user. For the single-user scenario, we derive a closed form for the optimal solution for both cases of point and extended targets. For the multi-user scenario, we show that both problems can be relaxed into semidefinite programming by using the semidefinite relaxation approach, and prove that the global optimum can always be obtained. Finally, we demonstrate numerically that the globally optimal solutions are reachable via the proposed methods, which provide significant gains in target estimation performance over state-of-the-art benchmarks.
\end{abstract}


%
\IEEEpeerreviewmaketitle

\section{Introduction}
\IEEEPARstart{W}{ireless} sensors and communication systems have shaped modern society in profound ways. 5G and beyond network is envisioned as an enabler for many emerging applications, such as intelligent connected vehicles and remote health-care. These applications demand wireless connectivity with tremendously increased data rates, substantially reduced latency, high-accuracy localization capability and support for massive devices \cite{8246850,8360863}. Indeed, in many location-aware services and applications, sensing and communications are recognized as a pair of intertwined functionalities, which are often required to operate simultaneously \cite{9127852,8828016}.
\\\indent To reduce costs and improve spectral-, energy-, and hardware-efficiency, the need for joint design of sensing and communication systems naturally arises in the above-mentioned scenarios \cite{8355705,8892631,7279172,7347464,7485066,9124713,8288677,8386661,9093221,ma2020spatial}. The integration between radar sensors and communication systems has received considerable attention from both industry and academia, motivating research on \emph{Dual-functional Radar-Communication (DFRC) Systems} \cite{8999605,7279172,7347464,7485066,7485316,9124713,8288677,8386661,9093221,ma2020spatial}. DFRC techniques combine both radar sensing and wireless communications via shared use of the spectrum, the hardware platform and a joint signal processing framework \cite{9127852,8999605}.
\\\indent One of the major challenges in DFRC is the design of a joint waveform, capable of the dual functionalities of target sensing and information delivering. The design methodology can be generally split into three categories: radar-centric design, communication-centric design, and joint design. Radar-centric approaches are built on the basis of a radar probing signal, which can be traced back to the early work \cite{mealey1963method}, where the communication data are modulated onto the radar pulses by pulse interval modulation (PIM). In this spirit, one may design a DFRC waveform by using a radar probing signal as an information carrier. Such examples include the combination of amplitude/phase shift keying (ASK/PSK) and linear frequency modulation (LFM) signals \cite{roberton2003integrated,saddik2007ultra}, as well as the marriage between spread spectrum sequences for communication and binary- and poly-phase codes for radar \cite{jamil2008integrated}.
\\\indent Communication-centric schemes rely on existing communication waveforms and standard-compatible protocols. For instance, the seminal work of \cite{sturm2011waveform} proposed employing orthogonal frequency division multiplexing (OFDM) for the use of target detection. The Doppler and delay processing for radar targets are decoupled in OFDM signals, and can be performed by the classical Fast Fourier Transform (FFT) and its inverse. The random impact of the communication data imposed on the target echo can be readily eliminated by element-wise division between the signal matrices of both the echo and the reference OFDM waveform. To further improve sensing performance, one may replace the sinusoidal carrier in the OFDM waveform with an LFM signal, which can be efficiently processed by leveraging the fractional Fourier transform (FrFT) \cite{7485314}. More recently, the IEEE 802.11ad protocol, which is a WLAN standard operating in the mmWave band, has been exploited to accomplish both radar sensing and communication tasks in vehicular networks \cite{8114253}.
\\\indent More relevant to this work, the DFRC waveform can be designed via jointly considering both functionalities, rather than based on existing radar or communication waveforms \cite{7279172,9124713}. Thanks to the higher degrees-of-freedom (DoFs) brought by multi-antenna arrays, these techniques have been well-explored recently from a spatial-domain perspective. Pioneered by \cite{7347464}, the authors proposed to embed the communication data by varying the sidelobes of the spatial beampattern of the MIMO radar, where the mainlobe is exploited solely for target detection. Again, the communication symbols can be expressed in the forms of various modulation formats, including ASK and PSK \cite{7347464,7485066}. However, such an approach supports only line-of-sight (LoS) communications, since the communication receivers must be within the correct sidelobe region in order to receive the symbols. To exploit both spatial and frequency diversity, one can embed communication data into the MIMO radar waveform through index or spatial modulation schemes \cite{9093221,ma2020spatial}, which are capable of tackling more complex non-LoS (NLoS) communication channels.
\\\indent The above-mentioned approaches are generally categorized as inter-pulse modulation schemes, which represent one communication codeword by a single radar pulse, and, consequently, result in low data rate tied to the pulse repetition frequency (PRF) of the radar. To improve communication performance, the work of \cite{8288677} proposed a beamforming design tailored to joint target sensing and multi-user NLoS communications, where the data symbols can be accommodated in each of the radar fast-time snapshots. As a step further, \cite{8386661} proposed a number of DFRC waveform optimization approaches, with minimizing the multi-user interference (MUI) as the objective function given specific radar constraints.
\\\indent To further exploit the spatial DoFs, massive MIMO arrarys, which are at the core of the 5G physical layer, have also been exploited for DFRC design. It has been shown by asymptotic analysis that when the size of the antenna array is sufficiently large, the target can be detected using only a single fast-time snapshot \cite{8962251}. Together with the use of mmWave and hybrid analog-digital beamforming techniques, the radar sensing functionality can be incorporated into the massive MIMO base station (BS) \cite{8999605,8550811,8827589}, which can be deployed in the 5G/B5G vehicular network as road side unit (RSU) for both vehicular communication and localization \cite{9171304}.
\\\indent While existing DFRC signalling strategies achieve favorable performance tradeoffs between radar and communications, they typically focus on the transmitter design, rather than on optimizing a bottom-line performance metric. More precisely, the target estimation performance for DFRC systems, which is characterized at the receiver side, is guaranteed implicitly by waveform shaping constraints, e.g., to approach some well-designed radar waveforms/beampatterns that are featured with good estimation capability under communication constraints \cite{9124713,8288677,8386661}. To the best of our knowledge, explicit optimization of estimation performance metrics has not been systematically studied in the context of DFRC design.
\begin{figure*}[ht]
\normalsize
\newcounter{MYtempeqncnt2}
\setcounter{MYtempeqncnt2}{\value{equation}}
\setcounter{equation}{3}
\begin{equation}\label{eq5}
\begin{gathered}
  \operatorname{CRB} \left( \theta  \right) = \frac{{\sigma _R^2\operatorname{tr} \left( {{{\mathbf{A}}^H}\left( \theta  \right){\mathbf{A}}\left( \theta  \right){{\mathbf{R}}_X}} \right)}}{{{{2\left| \alpha  \right|}^2}L\left( {\operatorname{tr} \left( {{{{\mathbf{\dot A}}}^H}\left( \theta  \right){\mathbf{\dot A}}\left( \theta  \right){{\mathbf{R}}_X}} \right)\operatorname{tr} \left( {{{\mathbf{A}}^H}\left( \theta  \right){\mathbf{A}}\left( \theta  \right){{\mathbf{R}}_X}} \right) - {{\left| {\operatorname{tr} \left( {{{{\mathbf{\dot A}}}^H}\left( \theta  \right){\mathbf{A}}\left( \theta  \right){{\mathbf{R}}_X}} \right)} \right|}^2}} \right)}}, \hfill \\
  \operatorname{CRB} \left( \alpha \right) = \frac{{\sigma _R^2\operatorname{tr} \left( {{{{\mathbf{\dot A}}}^H}\left( \theta  \right){\mathbf{\dot A}}\left( \theta  \right){{\mathbf{R}}_X}} \right)}}{{L\left( {\operatorname{tr} \left( {{{\mathbf{A}}^H}\left( \theta  \right){\mathbf{A}}\left( \theta  \right){{\mathbf{R}}_X}} \right)\operatorname{tr} \left( {{{{\mathbf{\dot A}}}^H}\left( \theta  \right){\mathbf{\dot A}}\left( \theta  \right){{\mathbf{R}}_X}} \right) - {{\left| {\operatorname{tr} \left( {{{{\mathbf{\dot A}}}^H}\left( \theta  \right){\mathbf{A}}\left( \theta  \right){{\mathbf{R}}_X}} \right)} \right|}^2}} \right)}}. \hfill \\
\end{gathered}
\end{equation}
\setcounter{equation}{\value{MYtempeqncnt2}}
\hrulefill
\vspace*{4pt}
\end{figure*}
\\\indent In this paper, we propose a design framework for multi-user MIMO (MU-MIMO) DFRC beamforming, with a specific emphasis on optimization of the target estimation performance, measured by the CRB for unbiased estimators \cite{kay1998fundamentals}. We consider two types of target models: point target and extended target, and derive the corresponding CRB expressions as functions of the beamforming matrix. We then formulate optimization problems to minimize the CRB, subject to individual SINR constraints for the users as well as a transmit power budget. Since the resulting formulations of the beamforming designs are non-convex, we propose semidefinite relaxation (SDR) approaches to solve them. To find the exact global optimum, we prove that rank-one solutions can be obtained for both problems. Finally, we provide numerical results to validate the performance of the proposed approaches. The simulations demonstrate that our designed beamformers significantly outperform the beampattern approximation based DFRC designs by reducing estimation errors.

We summarize our contributions as follows:
\begin{itemize}
  \item We propose a CRB-based framework for MU-MIMO DFRC beamforming design for both point and extended target estimation, with guaranteed SINRs for communication users.
  \item We analyze the structure of both problems when there is only a single communication user, and derive the optimal closed-form solution.
  \item We solve the multi-user DFRC beamforming problem in the point target scenario by the SDR approach \cite{5447068} when there are multiple users. Moreover, we prove that rank-one solutions can be obtained in general.
  \item We solve the multi-user DFRC beamforming problem in the extended target scenario also by the SDR approach, and prove that the rank-one globally optimal solution can be simply extracted in closed form.
\end{itemize}

The remainder of this paper is organized as follows: Section II introduces the system model. Section III derives the CRB for both point and extended target estimation. Section IV and V discuss the joint beamforming design for point and extended target scenarios. In Section VI we present numerical results, and finally Section VII concludes the paper.
\\\indent {\emph{Notations}}: Matrices are denoted by bold uppercase letters (i.e., $\mathbf{H}$), vectors are represented by bold lowercase letters (i.e., $\mathbf{w}$), and scalars are denoted by normal font (i.e., $L$); $\operatorname{tr}\left(\cdot\right)$ and $\operatorname{vec}\left(\cdot\right)$ denote the trace and the vectorization operations, $\left(\cdot\right)^T$, $\left(\cdot\right)^H$, and $\left(\cdot\right)^*$ stand for transpose, Hermitian transpose, and complex conjugate of the matrices. We use $\operatorname{Re}\left(\cdot\right)$ and $\operatorname{Im}\left(\cdot\right)$ to denote the real and imaginary parts of the argument. $l_2$ norm and the Frobenius norm are written as $\left\| \cdot\right\|$ and $\left\| \cdot\right\|_F$.
\section{System Model}
\subsection{System Setting}
We consider a MIMO DFRC BS equipped with $N_t$ transmit antennas and $N_r$ receive antennas, which is serving $K$ downlink single-antenna users while detecting a single target, as depicted in Fig. 1. Without loss of generality, we assume $K < N_t < N_r$.
\\\indent Let $\mathbf{X} \in \mathbb{C}^{N_t \times L}$ be a narrowband DFRC signal matrix, with $L > N_t$ being the length of the radar pulse/communication frame. From a communication perspective, $x_{i,j}$, i.e., the $\left(i,j\right)$-th entry of $\mathbf{X}$, represents the discrete signal sample transmitted at the $i$-th antenna and the $j$-th time slot. For the radar, $x_{i,j}$ is the $j$-th fast-time snapshot transmitted at the $i$-th antenna. In practice, $x_{i,j}$ needs to be associated with a sub-pulse or a pulse-shaping filter in order to formulate a continuous-time signal.
\\\indent By transmitting $\mathbf{X}$ to sense the target, the reflected echo signal matrix at the receiver DFRC BS is given by
\begin{equation}\label{eq1}
  {{\mathbf{Y}}_R} = {\mathbf{GX}} + {{\mathbf{Z}}_R},
\end{equation}
where ${{\mathbf{Z}}_R} \in \mathbb{C}^{N_r \times L}$ denotes an additive white Gaussian noise (AWGN) matrix, with variance of each entry being $\sigma_R^2$, and ${\mathbf{G}} \in {\mathbb{C}^{{N_r} \times {N_t}}}$ represents the target response matrix \cite{8579200}. The matrix $\mathbf{G}$ can be of different forms depending on the specific target models. In particular, we will consider the following two scenarios:
\begin{enumerate}
  \item \emph{Point target:} In this case, the target is located in the far field, e.g., a UAV that is far away from the BS, which can be thus viewed as a single point. The target response matrix can be written as
      \begin{equation}\label{eq2}
        {\mathbf{G}} = \alpha {\mathbf{b}}\left( \theta  \right){{\mathbf{a}}^H}\left( \theta  \right) \triangleq \alpha \mathbf{A}\left(\theta\right),
      \end{equation}
      where $\alpha \in \mathbb{C}$ represents the reflection coefficient, which contains both the round-trip path-loss and the radar cross-section (RCS) of the target, $\theta$ is the azimuth angle of the target relative to the BS, and finally $\mathbf{a}\left(\theta\right) \in \mathbb{C}^{N_t \times 1}$ and $\mathbf{b}\left(\theta\right) \in \mathbb{C}^{N_r \times 1}$ are steering vectors of the transmit and receive antennas, which are assumed to be a uniform linear array (ULA) with half-wavelength antenna spacing.
  \item \emph{Extended target:} In this case, the target is located in the near field, which is typically modeled as a surface with a large number of distributed point-like scatterers, such as a vehicle or a pedestrian moving on the road. Consequently, $\mathbf{G}$ can be further expressed as
      \begin{equation}\label{eq4}
        {\mathbf{G}} = \sum\limits_{i = 1}^{N_s} {{\alpha _i}{\mathbf{b}}\left( {{\theta _i}} \right){{\mathbf{a}}^H}\left( {{\theta _i}} \right)},
      \end{equation}
      where $N_s$ is the number of scatterers, $\alpha_i$ and $\theta_i$ denotes the reflection coefficient and the angle of the $i$-th scatterer. Note that due to the narrowband assumption, we consider only the angular spread of the target, and assume that all the point-like scatterers are located in the same range bin\footnote{While we focus on the single extended target here, we note that the model in (\ref{eq4}), as a generalization of (\ref{eq2}), can also represent multiple point targets in the far field.}. For extended targets, we are typically interested in estimating the response matrix $\mathbf{G}$ directly \cite{8579200}. Given $\mathbf{G}$, one can extract the angle and reflection coefficients of each point scatterer from the estimate of $\mathbf{G}$, using various approaches, such as MUSIC and APES algorithms \cite{506612,1143830}.
\end{enumerate}
Below we elaborate on the radar and communication performance metrics for both point and extended target scenarios. In particular, we rely on the CRB for target estimation, which is a lower bound on the variance of unbiased estimators, and employ the per-user SINR to measure the communication quality-of-service (QoS).
\begin{figure}[!t]
    \centering
    \includegraphics[width=0.7\columnwidth]{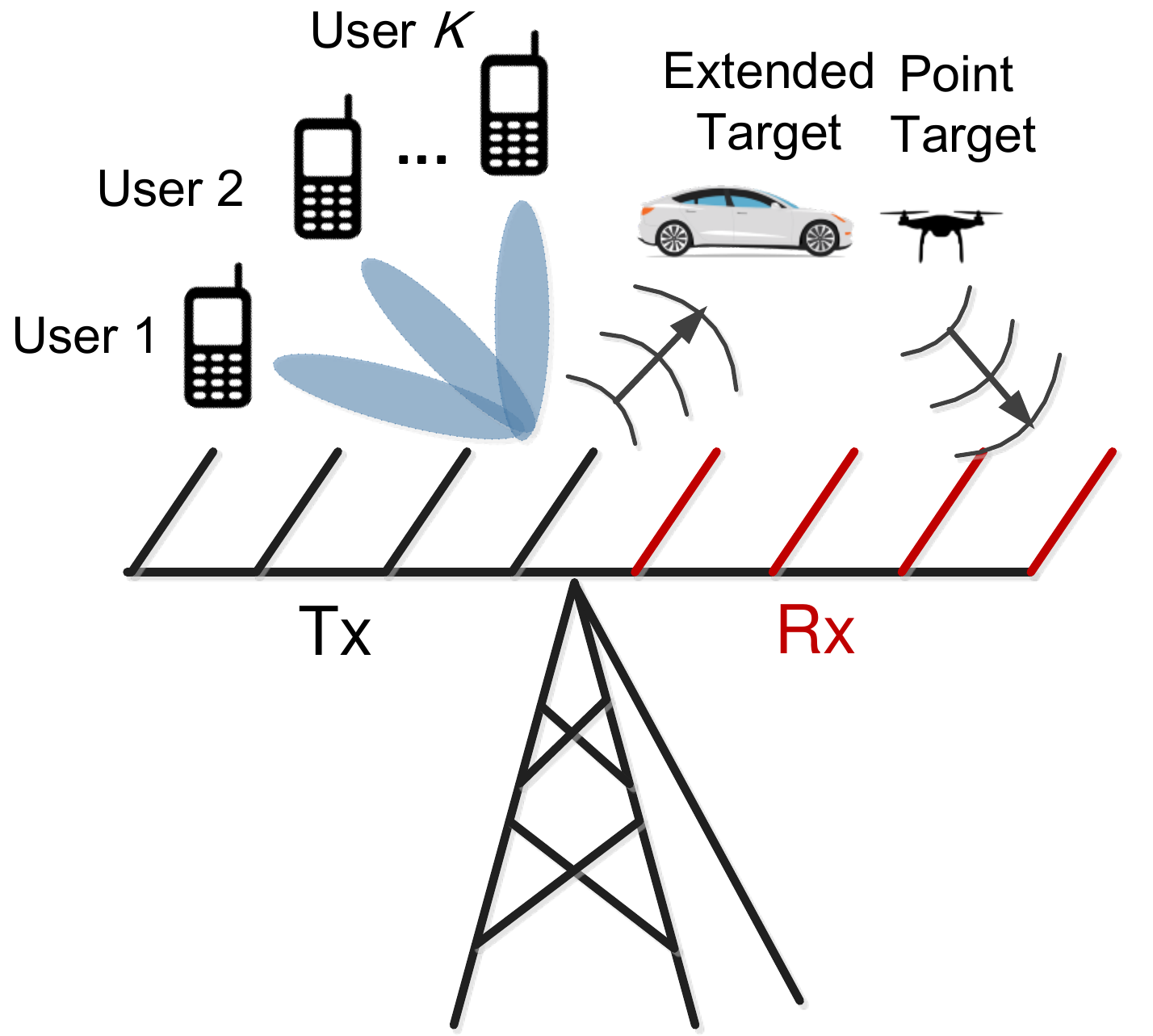}
    \caption{Dual-functional Radar-Communication System.}
    \label{fig:1}
\end{figure}
\subsection{Performance Metrics for the Point Target Case}
For the point target case, the CRBs for $\theta$ and $\alpha$ were derived in \cite{kay1998fundamentals,1703855}, and are given by (\ref{eq5}), where
\setcounter{equation}{4}
\begin{equation}\label{eq6}
  {{\mathbf{R}}_X} = \frac{1}{L}{\mathbf{X}}{{\mathbf{X}}^H} = \frac{1}{L}{\mathbf{W}}_D{{\mathbf{S}}_C}{\mathbf{S}}_C^H{{\mathbf{W}_D}^H} = {\mathbf{W}}_D{{\mathbf{W}}_D^H}
\end{equation}
is the sample covariance matrix of $\mathbf{X}$, and ${\mathbf{\dot A}}\left( \theta  \right) = \frac{{\partial {\mathbf{A}}\left( \theta  \right)}}{{\partial \theta }}$.
\\\indent By transmitting $\mathbf{X}$ to $K$ users, the received signal matrix at the communication receivers is
\begin{equation}\label{eq7}
  {{\mathbf{Y}}_C} = {\mathbf{HX}} + {{\mathbf{Z}}_C},
\end{equation}
where ${{\mathbf{Z}}_C} \in \mathbb{C}^{K \times L}$ is an AWGN matrix with the variance of each entry being $\sigma_C^2$, and ${\mathbf{H}} = {\left[ {{{\mathbf{h}}_1},{{\mathbf{h}}_2}, \ldots ,{{\mathbf{h}}_K}} \right]^H} \in \mathbb{C}^{K \times N_t}$ represents the communication channel matrix, which is assumed to be known to the BS, with each entry being independently distributed. The matrix $\mathbf{X}$ is given by
\begin{equation}\label{eq8}
  {\mathbf{X}} = {\mathbf{W}}_D{\mathbf{S}}_C,
\end{equation}
where $\mathbf{W}_D$ is the dual-functional beamforming matrix to be designed, and ${\mathbf{S}}_C \in \mathbb{C}^{K \times L}$ contains $K$ unit-power data streams intended for the $K$ users. The data streams are assumed to be orthogonal to each other so that
\begin{equation}\label{eq9}
  \frac{1}{L}{{\mathbf{S}}_C}{\mathbf{S}}_C^H = {{\mathbf{I}}_K}.
\end{equation}

By denoting the beamforming matrix as ${\mathbf{W}}_D = \left[ {{{\mathbf{w}}_1},{{\mathbf{w}}_2}, \ldots ,{{\mathbf{w}}_K}} \right]$, with the $k$-th column being the beamformer for the $k$-th user, its SINR is given as
\begin{equation}\label{eq10}
  {\gamma _k} = \frac{{{{\left| {{\mathbf{h}}_k^H{{\mathbf{w}}_k}} \right|}^2}}}{{\sum\nolimits_{i = 1,i \ne k}^K {{{\left| {{\mathbf{h}}_k^H{{\mathbf{w}}_i}} \right|}^2} + \sigma _C^2} }}.
\end{equation}

\subsection{Performance Metrics for the Extended Target Case}
For the extended target case, (\ref{eq1}) is nothing but a linear white Gaussian model in $\mathbf{G}$, where the Fisher Information Matrix (FIM) with respect to ${\mathbf{G}}$ is known to be \cite{4838872,kay1998fundamentals}
\begin{equation}\label{eq11}
{\mathbf{J}} = \frac{1}{{\sigma _R^2{N_r}}}{\mathbf{X}}{{\mathbf{X}}^H} = \frac{L}{{\sigma _R^2{N_r}}}{{\mathbf{R}}_X}.
\end{equation}
\\\indent By examining (\ref{eq8}) we see that $\mathbf{X} \in \mathbb{C}^{N_t \times L}$ is rank-deficient, since
\begin{equation}\label{eq12}
\operatorname{rank} \left( {\mathbf{X}} \right) \le \min \left\{ {\operatorname{rank} \left( {\mathbf{W}}_D \right),\operatorname{rank} \left( {\mathbf{S}_C} \right)} \right\} = K < N_t \le L.
\end{equation}
Consequently, if we transmit only $K$ signal streams, the DoFs available are not enough to recover the rank-$N_t$ matrix $\mathbf{G}$. Moreover, the FIM $\mathbf{J}$ becomes singular, resulting in the non-existence of the unbiased estimator according to \cite{890346,4838872}. Note that this is not an issue in the point-target scenario, in which case $K$ DoFs are more than enough to estimate $\theta$ and $\alpha$. For the extended target scenario, one may constrain $\mathbf{G}$ into some subset, and employ a modified CRB instead \cite{4838872}. This, however, leads to inevitable performance loss for target estimation, due to the lack of radar DoFs. In order to guarantee a satisfactory radar performance, we propose introducing an extra structure to matrix $\mathbf{X}$ for extending the DoFs to its maximum, i.e., $N_t$, by transmitting dedicated probing streams in addition to data streams intended for $K$ users. Note that these signal streams are dedicated to target probing, without carrying communication data. Let us consider an augmented data matrix, given as
\begin{equation}\label{eq13}
  {\mathbf{\tilde S}} = \left[ \begin{gathered}
  {\mathbf{S}_C} \hfill \\
  {{\mathbf{S}}_A} \hfill \\
\end{gathered}  \right] \in {\mathbb{C}^{{\left(K+N_t\right)} \times L}},
\end{equation}
where ${{\mathbf{S}}_A} \in {\mathbb{C}^{{{N_t}}  \times L}}$ denotes the dedicated probing streams, and is orthogonal to $\mathbf{S}_C$. Therefore, it still holds true that
\begin{equation}\label{eq14}
  \frac{1}{L}{\mathbf{\tilde S}}{{{\mathbf{\tilde S}}}^H} = {{\mathbf{I}}_{{K+N_t}}}.
\end{equation}
We further augment the beamforming matrix in the form of
\begin{equation}\label{eq15}
 {{\mathbf{\tilde W}}_D} = \left[ {{{\mathbf{w}}_1},{{\mathbf{w}}_2}, \ldots, {{\mathbf{w}}_{{K+N_t}}}} \right] = \left[ {{{\mathbf{W}}_C},{{\mathbf{W}}_A}} \right] \in {\mathbb{C}^{{\left(K+N_t\right)} \times {N_t}}},
\end{equation}
where ${\mathbf{W}_C} = \left[ {{{\mathbf{w}}_1}, \ldots, {{\mathbf{w}}_{{K}}}} \right] \in \mathbb{C}^{N_t \times K}$ is the communication beamformer, and ${{\mathbf{W}}_A} = \left[ {{{\mathbf{w}}_{K+1}},\ldots, {{\mathbf{w}}_{{K+N_t}}}} \right] \in \mathbb{C}^{N_t \times N_t}$ is the auxiliary beamforming matrix for the probing streams. By properly designing $\mathbf{\tilde W}_D$, the resulting transmitted signal matrix $\mathbf{X} = {\mathbf{\tilde W}}_D{\mathbf{\tilde S}}$ will have a full rank of $N_t$. Note that in this DFRC signal model, the overall beamformer ${\mathbf{\tilde W}}_D$ is used for sensing the extended target, for guaranteeing the estimation performance and the feasibility of unbiased estimation. The first $K$ columns of ${\mathbf{\tilde W}}_D$, i.e., ${\mathbf{W}_C}$, convey information data to $K$ users.
\\\indent Based on the above, the sample covariance matrix of $\mathbf{X}$ is then given by
\begin{equation}\label{eq16}
  {{\mathbf{R}}_X} = {{\mathbf{\tilde W}}_D}{\mathbf{\tilde W}}_D^H = {{\mathbf{W}}_C}{\mathbf{W}}_C^H + {{\mathbf{W}}_A}{\mathbf{W}}_A^H,
\end{equation}
which has a full rank of $N_t$ and is now invertible. Therefore, the CRB for $\mathbf{G}$ can be expressed as
\begin{equation}\label{eq17}
\operatorname{CRB} \left( {\mathbf{G}} \right) = \operatorname{tr} \left( {{{\mathbf{J}}^{ - 1}}} \right) = \frac{{\sigma _R^2{N_r}}}{L}\operatorname{tr} \left( {{\mathbf{R}}_X^{ - 1}} \right).
\end{equation}
The CRB above is achievable using maximum likelihood estimation (MLE). This is because the MLE of $\mathbf{G}$ is simply a linear estimation problem in the presence of the i.i.d. Gaussian noise, whose mean squared error (MSE) equals the CRB.
\\\indent Note that the dedicated probing signals would impose interference on the communication users, as $\mathbf{S}_A$ does not contain any useful information. The per-user SINR expression for the extended target case should therefore be modified as
\begin{equation}\label{eq18}
{{\tilde\gamma} _k} = \frac{{{{\left| {{\mathbf{h}}_k^H{{\mathbf{w}}_k}} \right|}^2}}}{{\sum\nolimits_{i = 1,i \ne k}^K {{{\left| {{\mathbf{h}}_k^H{{\mathbf{w}}_i}} \right|}^2} + {{\left\| {{\mathbf{h}}_k^H{{\mathbf{W}}_A}} \right\|}^2} + \sigma _C^2} }},
\end{equation}
where the radar interference is imposed in (\ref{eq18}) as part of the denominator.

\section{Joint Beamforming Design for Point Target}
\subsection{Problem Formulation}
For the point target, the beamforming optimization problem under communication user's SINR and power budget constraints is formulated as
\begin{equation}\label{eq19}
\begin{gathered}
  \mathop {\min }\limits_{{\mathbf{W}}_D} \;\operatorname{CRB} \left( \theta  \right) \hfill \\
  \text{\text{\text{s.t.}}}\;\;{{\gamma }_k} \ge {\Gamma _k},\forall k, \hfill \\
  \;\;\;\;\;\left\| {{\mathbf{W}}_D} \right\|_F^2 \le {P_T}, \hfill \\
\end{gathered}
\end{equation}
where $\Gamma_k$ is the required SINR for the $k$-th user, and $P_T$ is the transmit power budget. For notational convenience, we will not elaborate on the minimization of $\operatorname{CRB} \left( \alpha  \right)$ in the second line of (\ref{eq5}), since its form is similar to that of $\operatorname{CRB} \left( \theta  \right)$. While $\operatorname{CRB} \left( \theta  \right)$ relies on the value of $\theta$ as can be seen in (\ref{eq5}), the above problem can be interpreted as optimizing ${\mathbf{W}}_D$ with respect to a direction of interest, where there might be a potential target with azimuth angle of $\theta$. This is quite typical in the target tracking scenario where the radar wishes to beamform towards an estimated/predicted direction to track the movement of the target. Thus, we assume $\alpha$ is known, and incorporate it into the radar receive SNR.
\\\indent In what follows, we analyze problem (\ref{eq19}) under both single- and multi-user scenarios.
\subsection{Single-User Case}
In this subsection, we provide a closed-form solution to problem (\ref{eq19}) when there is only a single user. First of all, let us choose the center of the ULA antennas as the reference point, in which case the transmit steering vector and its derivative can be written as (assuming even number of antennas) \cite{1703855}
\begin{equation}\label{eq20}
  {\mathbf{a}}\left( \theta  \right) = {\left[ {{e^{ - j\frac{{{N_t-1}}}{2}\pi \sin \theta }},{e^{ - j\frac{{{N_t} - 3}}{2}\pi \sin \theta }}, \ldots ,{e^{j\frac{{{N_t-1}}}{2}\pi \sin \theta }}} \right]^T},
\end{equation}
\begin{equation}\label{eq21}
{\mathbf{\dot a}}\left( \theta  \right) = {\left[ { - j{a_1}\frac{{{N_t-1}}}{2}\pi \cos \theta , \ldots ,j{a_{{N_t}}}\frac{{{N_t-1}}}{2}\pi \cos \theta } \right]^T},
\end{equation}
where $a_i$ represents the $i$-th entry of ${\mathbf{a}}\left( \theta  \right)$. The receive steering vector and its derivative take similar forms of (\ref{eq20}) and (\ref{eq21}). It can be easily verified that due to the symmetry,
\begin{equation}\label{eq22}
  {{\mathbf{a}}^H}{\mathbf{\dot a}} = 0,{{\mathbf{b}}^H}{\mathbf{\dot b}} = 0, \forall \theta,
\end{equation}
where ${{\mathbf{a}}}$, ${{\mathbf{b}}}$, ${{\mathbf{\dot a}}}$, and ${{\mathbf{\dot b}}}$ denote ${{\mathbf{a}}} \left(\theta\right)$, ${{\mathbf{b}}}\left(\theta\right)$, ${{\mathbf{\dot a}}}\left(\theta\right)$, and ${{\mathbf{\dot b}}}\left(\theta\right)$, respectively. By letting $\mathbf{w}_1$, $\mathbf{h}_1$, and $\Gamma_1$ be the beamforming vector, the channel and the required SINR threshold in the single-user case, we have ${{\mathbf{R}}_X} = {{\mathbf{w}}_1}{\mathbf{w}}_1^H$. Leveraging the orthogonality property (\ref{eq22}) yields
\begin{equation}\label{eq23}
\begin{gathered}
  \operatorname{tr} \left( {{{\mathbf{A}}^H}{\mathbf{A}}{{\mathbf{R}}_X}} \right) = \operatorname{tr} \left( {{\mathbf{b}}{{\mathbf{a}}^H}{{\mathbf{w}}_1}{\mathbf{w}}_1^H{\mathbf{a}}{{\mathbf{b}}^H}} \right) = {\left\| {\mathbf{b}} \right\|^2}{\left| {{{\mathbf{a}}^H}{{\mathbf{w}}_1}} \right|^2}, \hfill \\
  \operatorname{tr} \left( {{{{\mathbf{\dot A}}}^H}{\mathbf{A}}{{\mathbf{R}}_X}} \right) = \operatorname{tr} \left( {{\mathbf{b}}{{\mathbf{a}}^H}{{\mathbf{w}}_1}{\mathbf{w}}_1^H\left( {{\mathbf{a}}{{{\mathbf{\dot b}}}^H} + {\mathbf{\dot a}}{{\mathbf{b}}^H}} \right)} \right)\hfill \\
   = {\left\| {\mathbf{b}} \right\|^2}{{\mathbf{a}}^H}{{\mathbf{w}}_1}{\mathbf{w}}_1^H{\mathbf{\dot a}}, \hfill \\
  \operatorname{tr} \left( {{{{\mathbf{\dot A}}}^H}{\mathbf{\dot A}}{{\mathbf{R}}_X}} \right) = \operatorname{tr} \left( {\left( {{\mathbf{\dot b}}{{\mathbf{a}}^H} + {\mathbf{b}}{{{\mathbf{\dot a}}}^H}} \right){{\mathbf{w}}_1}{\mathbf{w}}_1^H\left( {{\mathbf{a}}{{{\mathbf{\dot b}}}^H} + {\mathbf{\dot a}}{{\mathbf{b}}^H}} \right)} \right) \hfill \\
   = {\left\| {{\mathbf{\dot b}}} \right\|^2}{\left| {{{\mathbf{a}}^H}{{\mathbf{w}}_1}} \right|^2} + {\left\| {\mathbf{b}} \right\|^2}{\left| {{{{\mathbf{\dot a}}}^H}{{\mathbf{w}}_1}} \right|^2}, \hfill \\
\end{gathered}
\end{equation}
where ${\mathbf{A}} \triangleq {\mathbf{A}}\left( \theta  \right),{\mathbf{\dot A}} \triangleq {\mathbf{\dot A}}\left( \theta  \right)$. Substituting (\ref{eq23}) into (\ref{eq15}), $\operatorname{CRB}\left(\theta\right)$ can be simplified as
\begin{equation}\label{eq24}
  \operatorname{CRB} \left( \theta  \right) = \frac{{\sigma _R^2\left\| {{\mathbf{\dot b}}} \right\|{{\left| {{{\mathbf{a}}^H}{{\mathbf{w}}_1}} \right|}^2}}}{{2{{\left| \alpha  \right|}^2}L}}.
\end{equation}
Accordingly, the optimization problem (\ref{eq19}) can be recast as
\begin{equation}\label{eq25}
  \begin{gathered}
  \mathop {\max }\limits_{\mathbf{w}_1} \;{\left| {{{\mathbf{a}}^H}{{\mathbf{w}}_1}} \right|^2}  \hfill \\
  \text{\text{s.t.}}\;\;\;{\left| {{\mathbf{h}}_1^H{{\mathbf{w}}_1}} \right|^2} \ge \Gamma_1\sigma _C^2,\;{\left\| {{{\mathbf{w}}_1}} \right\|^2} \le {P_T}. \hfill \\
\end{gathered}
\end{equation}
Note that in the single-user case, the CRB minimization problem reduces to maximizing the radiation power at angle $\theta$. We next prove the following lemma.
\begin{lemma}
  The optimal solution of (\ref{eq25}) satisfies
  \begin{equation}\label{eq26}
    {{\mathbf{w}}_1} \in \operatorname{span}\left\{ {{\mathbf{a}},{{\mathbf{h}}_1}} \right\}.
  \end{equation}
\end{lemma}
\renewcommand{\qedsymbol}{$\blacksquare$}
\begin{proof}
See Appendix A.
\end{proof}
Using Lemma 1, the optimal solution is given by the following theorem.
\begin{thm}
  The optimal solution to (\ref{eq25}) is
  \begin{equation}\label{eq27}
  {{\mathbf{w}}_1} = \left\{ \begin{gathered}
  \sqrt {{P_T}} \frac{{\mathbf{a}}}{{\left\| {\mathbf{a}} \right\|}},\;\;{\text{if}}\;{P_T}{\left| {{\mathbf{h}}_1^H{\mathbf{a}}} \right|^2} > {N_t}\Gamma_1 \sigma _C^2, \hfill \\
  {x_1}{{\mathbf{u}}_1} + {x_2}{{\mathbf{a}}_u},\;{\text{otherwise}}, \hfill \\
\end{gathered}  \right.
  \end{equation}
  where
  \begin{equation}\label{eq28}
    {{\mathbf{u}}_1} = \frac{{{{\mathbf{h}}_1}}}{{\left\| {{{\mathbf{h}}_1}} \right\|}},\quad {{\mathbf{a}}_u} = \frac{{{\mathbf{a}} - \left( {{\mathbf{u}}_1^H{\mathbf{a}}} \right){{\mathbf{u}}_1}}}{{\left\| {{\mathbf{a}} - \left( {{\mathbf{u}}_1^H{\mathbf{a}}} \right){{\mathbf{u}}_1}} \right\|}},
  \end{equation}
  \begin{equation}\label{eq29}
    {x_1} = \sqrt {\frac{{{\Gamma _1}\sigma _C^2}}{{{{\left\| {{{\mathbf{h}}_1}} \right\|}^2}}}} \frac{{{\mathbf{u}}_1^H{\mathbf{a}}}}{{\left| {{\mathbf{u}}_1^H{\mathbf{a}}} \right|}},\quad {x_2} = \sqrt {{P_T} - \frac{{{\Gamma _1}\sigma _C^2}}{{{{\left\| {{{\mathbf{h}}_1}} \right\|}^2}}}} \frac{{{\mathbf{a}}_u^H{\mathbf{a}}}}{{\left| {{\mathbf{a}}_u^H{\mathbf{a}}} \right|}}.
  \end{equation}
\end{thm}
\begin{proof}
See Appendix B.
\end{proof}
\subsection{Semidefinite Relaxation for the Multi-User Case}
By taking a closer look at (\ref{eq5}), we see that $\operatorname{CRB} \left( \theta  \right)$ as an objective function is non-convex in $\mathbf{R}_X$ due to its fractional structure. Fortunately, it can be equivalently transformed into a convex expression with respect to $\mathbf{R}_X$ by relying on the following proposition.
\begin{proposition} \label{prop1}
  Minimizing $\operatorname{CRB} \left( \theta  \right)$ is equivalent to solving the following SDP
  \begin{equation}\label{eq30}
    \begin{gathered}
    \mathop {\min }\limits_{{{\mathbf{R}}_X}\succeq {\mathbf{0}},t} \; - t \hfill \\
    \text{\text{s.t.}}\;\;\left[ {\begin{array}{*{20}{c}}
    {\operatorname{tr} \left( {{{{\mathbf{\dot A}}}^H}{\mathbf{\dot A}}{{\mathbf{R}}_X}} \right) - t}&{\operatorname{tr} \left( {{{{\mathbf{\dot A}}}^H}{\mathbf{A}}{{\mathbf{R}}_X}} \right)} \\
    {\operatorname{tr} \left( {{{\mathbf{A}}^H}{\mathbf{\dot A}}{{\mathbf{R}}_X}} \right)}&{\operatorname{tr} \left( {{{\mathbf{A}}^H}{\mathbf{A}}{{\mathbf{R}}_X}} \right)}
    \end{array}} \right] \succeq {\mathbf{0}}. \hfill \\
    \end{gathered}
  \end{equation}
\end{proposition}
\renewcommand{\qedsymbol}{$\blacksquare$}
\begin{proof}
See Appendix C.
\end{proof}
Based on Proposition 1, and by noting that ${{\mathbf{R}}_X} = {\mathbf{W}}_D{{{\mathbf{W}}}_D^H} = \sum\nolimits_{k = 1}^{{K}} {{{\mathbf{w}}_k}{\mathbf{w}}_k^H}$, problem (\ref{eq19}) can be rewritten as
\begin{equation}\label{eq31}
\begin{gathered}
  \mathop {\min }\limits_{\left\{ {{{\mathbf{w}}_k}} \right\}_{k = 1}^K,{{\mathbf{R}}_X},t} \; - t \hfill \\
  {\text{s.t.}}\;\left[ {\begin{array}{*{20}{c}}
  {\operatorname{tr} \left( {{{{\mathbf{\dot A}}}^H}{\mathbf{\dot A}}{{\mathbf{R}}_X}} \right) - t}&{\operatorname{tr} \left( {{{{\mathbf{\dot A}}}^H}{\mathbf{A}}{{\mathbf{R}}_X}} \right)} \\
  {\operatorname{tr} \left( {{{\mathbf{A}}^H}{\mathbf{\dot A}}{{\mathbf{R}}_X}} \right)}&{\operatorname{tr} \left( {{{\mathbf{A}}^H}{\mathbf{A}}{{\mathbf{R}}_X}} \right)}
\end{array}} \right] \succeq {\mathbf{0}}, \hfill \\
  \;\;\;\;\;\;\frac{{{{\left| {{\mathbf{h}}_k^H{{\mathbf{w}}_k}} \right|}^2}}}{{\sum\nolimits_{i = 1,i \ne k}^K {{{\left| {{\mathbf{h}}_k^H{{\mathbf{w}}_i}} \right|}^2} +\sigma _C^2} }} \ge {\Gamma _k},\forall k, \hfill \\
  \;\;\;\;\;\;\sum\nolimits_{k = 1}^{K} {\operatorname{tr} \left( {{{\mathbf{w}}_k}{\mathbf{w}}_k^H} \right)} \le {P_T},\quad {{\mathbf{R}}_X} = \sum\nolimits_{k = 1}^{{K}} {{{\mathbf{w}}_k}{\mathbf{w}}_k^H}. \hfill \\
\end{gathered}
\end{equation}
While problem (\ref{eq31}) is still non-convex, it can be relaxed into a convex problem by using the classical SDR technique. Upon letting ${{\mathbf{Q}}_k} = {\mathbf{h}}_k{\mathbf{h}}_k^H, {{\mathbf{W}}_k} = {{\mathbf{w}}_k}{\mathbf{w}}_k^H$, the $k$-th SINR constraint can be reformulated as
\begin{equation}\label{eq33}
\operatorname{tr} \left( {{{\mathbf{Q}}_k}{{\mathbf{W}}_k}} \right) - {\Gamma _k}\sum\nolimits_{i = 1,i \ne k}^{K} {\operatorname{tr} \left( {{{\mathbf{Q}}_k}{{\mathbf{W}}_i}} \right)} \ge {\Gamma _k}\sigma _C^2.
\end{equation}
Note that the desired solution requires $\operatorname{rank}\left(\mathbf{W}_k\right) = 1$, and $\mathbf{W}_k \succeq \mathbf{0}$. By droping the rank constraints on $\mathbf{W}_k,\forall k$, problem (\ref{eq31}) can be relaxed as
\begin{equation}\label{eq34}
\begin{gathered}
  \mathop {\min }\limits_{\left\{ {{{\mathbf{W}}_k}} \right\}_{k = 1}^{K},t} \; - t \hfill \\
  \text{\text{s.t.}}\left[ {\begin{array}{*{20}{c}}
  {\operatorname{tr} \left( {{{{\mathbf{\dot A}}}^H}{\mathbf{\dot A}}\sum\limits_{k = 1}^K {{{\mathbf{W}}_k}} } \right) - t}&{\operatorname{tr} \left( {{{{\mathbf{\dot A}}}^H}{\mathbf{A}}\sum\limits_{k = 1}^K {{{\mathbf{W}}_k}} } \right)} \\
  {\operatorname{tr} \left( {{{\mathbf{A}}^H}{\mathbf{\dot A}}\sum\limits_{k = 1}^K {{{\mathbf{W}}_k}} } \right)}&{\operatorname{tr} \left( {{{\mathbf{A}}^H}{\mathbf{A}}\sum\limits_{k = 1}^K {{{\mathbf{W}}_k}} } \right)}
\end{array}} \right] \succeq {\mathbf{0}},\; \hfill \\
  \;\;\;\;\;\;\operatorname{tr} \left( {{{\mathbf{Q}}_k}{{\mathbf{W}}_k}} \right) - {\Gamma _k}\sum\nolimits_{i = 1,i \ne k}^{K} {\operatorname{tr} \left( {{{\mathbf{Q}}_k}{{\mathbf{W}}_i}} \right)} \ge {\Gamma _k}\sigma _C^2,\forall k, \hfill \\
  \;\;\;\;\;\;\sum\nolimits_{k = 1}^{K} {\operatorname{tr} \left( {{{\mathbf{W}}_k}} \right)}  \le {P_T},{{\mathbf{W}}_k} \succeq {\mathbf{0}},\forall k, \hfill \\
\end{gathered}
\end{equation}
which is a standard SDP and can be solved via numerical tools like CVX \cite{grant2008cvx}. To show the achievability of rank-one optimal solutions, we provide further insights into problem (\ref{eq34}) by proving the following theorem.
\begin{thm}
  Suppose that problem (\ref{eq34}) is feasible. Let ${\mathbf{\bar A}} = \left[ {{\mathbf{a}},{\mathbf{\dot a}}} \right]$. Under the condition that ${\mathbf{H}}{\mathbf{\bar A}}$ is of full column rank, the optimal solution $\left\{ {{{{\mathbf{W}}}_k}} \right\}_{k = 1}^{{K}}$ always satisfies $\operatorname{rank}\left({{{\mathbf{W}}}_k}\right) = 1, \forall k$.
\end{thm}
\begin{proof}
  See Appendix D.
\end{proof}
\indent\emph{Remark:} Note that as we consider a random channel $\mathbf{H}$ whose entries are independently distributed, the condition that ${\mathbf{H}}{\mathbf{\bar A}}$ is of full column rank almost always holds for $K \ge 2$. Therefore, solving (\ref{eq34}) yields rank-one solutions in general, i.e., the globally optimal solutions of problem (\ref{eq31}). In other words, if the communication channel and the radar target channel are not correlated with each other, the optimal beamformer can always be obtained by solving the SDR problem (\ref{eq34}).

\section{Joint Beamforming Design for Extended Target}
\subsection{Problem Formulation}
Based on the discussion in Sec. II-B, as well as (\ref{eq16}) and (\ref{eq17}), the beamforming optimization problem for the extended target scenario can be expressed as
\begin{equation}\label{eq35}
\begin{gathered}
  \mathop {\min }\limits_{{\mathbf{\tilde W}}_D} \;\operatorname{CRB}\left(\mathbf{G}\right) = \operatorname{tr} \left( {{{\left( {{{\mathbf{W}}_C}{\mathbf{W}}_C^H + {{\mathbf{W}}_A}{\mathbf{W}}_A^H} \right)}^{ - 1}}} \right) \hfill \\
  \text{\text{s.t.}}\;\;{\tilde\gamma}_k \ge {\Gamma _k},k = 1,\ldots,K, 
  \;\;\;\;\;\;\left\| {{\mathbf{\tilde W}}_D} \right\|_F^2 \le {P_T},\; \hfill \\
\end{gathered}
\end{equation}
By letting
\begin{equation}\label{eq36}
  {{\mathbf{W}}_k} = {{\mathbf{w}}_k}{\mathbf{w}}_k^H,\forall k \le K,\quad {{\mathbf{W}}_{K + 1}} = {{\mathbf{W}}_A}{\mathbf{W}}_A^H,
\end{equation}
and following similar steps of its point target counterpart, (\ref{eq35}) can be relaxed to the following convex form
\begin{equation}\label{eq37}
\begin{gathered}
  \mathop {\min }\limits_{\left\{ {{{\mathbf{W}}_k}} \right\}_{k = 1}^{K + 1}} \;\;\operatorname{tr} \left( {{{\left( {\sum\nolimits_{k = 1}^{K + 1} {{{\mathbf{W}}_k}} } \right)}^{ - 1}}} \right) \hfill \\
  \;\;\;\;{\text{s.t.}}\;\;\operatorname{tr} \left( {{{\mathbf{Q}}_k}{{\mathbf{W}}_k}} \right) - {\Gamma _k}\sum\nolimits_{i = 1,i \ne k}^{K+1} {\operatorname{tr} \left( {{{\mathbf{Q}}_k}{{\mathbf{W}}_i}} \right)} \ge {\Gamma _k}\sigma _C^2,\forall k, \hfill \\
  \;\;\;\;\;\;\;\;\;\;\;\sum\nolimits_{k = 1}^{K + 1} {\operatorname{tr} \left( {{{\mathbf{W}}_k}} \right)}  \le {P_T},{{\mathbf{W}}_k} \succeq {\mathbf{0}},\forall k. \hfill \\
\end{gathered}
\end{equation}
Next, we show that problem (\ref{eq37}) can be solved in closed form when there is only a single communication user to be served.

\subsection{Single-User Case}
Let us denote the SINR threshold, the covariance matrices of the channel vector, and the beamformer for the user as $\Gamma_1$, $\mathbf{Q}_1 = {\mathbf{h}}_1{\mathbf{h}}_1^H$, and $\mathbf{W}_1 = {\mathbf{w}}_1{\mathbf{w}}_1^H$, respectively. The optimization problem (\ref{eq37}) can be recast as
\begin{equation}\label{eq38}
\begin{gathered}
  \mathop {\min }\limits_{{{\mathbf{W}}_1},{{\mathbf{R}}_X}} \;\operatorname{tr} \left( {{\mathbf{R}}_X^{ - 1}} \right) \hfill \\
  \;\;\;\text{\text{s.t.}}\;\;\operatorname{tr} \left( {{{\mathbf{Q}}_1}{{\mathbf{W}}_1}} \right) - {\Gamma _1}\operatorname{tr} \left( {{{\mathbf{Q}}_1}\left( {{{\mathbf{R}}_X} - {{\mathbf{W}}_1}} \right)} \right) \ge {\Gamma _1}\sigma _C^2, \hfill \\
  \;\;\;\;\;\;\;\;\;\;\operatorname{tr} \left( {{{\mathbf{R}}_X}} \right) \le {P_T}, {{\mathbf{R}}_X} \succeq {{\mathbf{W}}_1} \succeq {\mathbf{0}}. \hfill \\
\end{gathered}
\end{equation}
In this case we have ${{\mathbf{R}}_X} = \mathbf{W}_1 + \mathbf{W}_2$, where $\mathbf{W}_2 = {{\mathbf{W}}_A}{\mathbf{W}}_A^H$. Let us compute the eigenvalue decomposition of ${{\mathbf{Q}}_1}$ as ${{\mathbf{Q}}_1} = {{\mathbf{U}}}{{\mathbf{\Sigma }}}{\mathbf{U}}^H$, where ${\mathbf{U}} = \left[ {{{\mathbf{u}}_1},{{\mathbf{u}}_2}, \ldots ,{{\mathbf{u}}_{{N_t}}}} \right]$ contains the eigenvectors of $\mathbf{Q}_1$ with ${{\mathbf{u}}_1} = \frac{{{\mathbf{h}}_1}}{{\left\| {{{\mathbf{h}}_1}} \right\|}}$, and $\mathbf{\Sigma }$ is a diagonal matrix with only a single non-zero eigenvalue ${\left\| {{{\mathbf{h}}_1}} \right\|^2}$.
\begin{lemma}
  The optimal ${{{\mathbf{R}}_X}}$ of (\ref{eq38}) can be written as
  \begin{equation}\label{eq39}
    {{\mathbf{R}}_X} = {{\mathbf{U}}}{\mathbf{\Lambda}}{\mathbf{U}}^H,
  \end{equation}
  where $\mathbf{\Lambda}$ is a full-rank real diagonal matrix. In other words, ${{\mathbf{u}}_1},\dots,{{\mathbf{u}}_{N_t}}$ are also the eigenvectors of the optimal ${{{\mathbf{R}}_X}}$.
\end{lemma}
\begin{proof}
Suppose that (\ref{eq38}) is feasible. Given an optimal $\mathbf{R}_X$ that reaches the minimum, one can always construct an optimal $\mathbf{W}_1$ in the form of
\begin{equation}\label{eq40}
{{\mathbf{W}}_1} = \left( {{\mathbf{u}}_1^H{{\mathbf{R}}_X}{{\mathbf{u}}_1}} \right){{\mathbf{u}}_1}{\mathbf{u}}_1^H.
\end{equation}
This is because by (\ref{eq40}) we project ${{\mathbf{R}}_X}$ onto the direction of $\mathbf{h}_1$, in which case $\operatorname{tr} \left( {{{\mathbf{Q}}_1}{{\mathbf{W}}_1}} \right)$ is maximized, and the interference term is minimized to zero. Accordingly, the SINR is maximized, while the objective value and the transmit power remain unchanged. Therefore, (\ref{eq40}) is a solution to (\ref{eq38}).
\\\indent It can be immediately observed that the optimal $\mathbf{W}_2$ is orthogonal to $\mathbf{W}_1$. By further noting that $\operatorname{rank}\left(\mathbf{W}_2\right) \ge N_t - 1$, it should have the form of ${{\mathbf{W}}_2} = \sum\nolimits_{i = 2}^{{N_t}} {{\lambda _{ii}}} {{\mathbf{u}}_i}{\mathbf{u}}_i^H$, with $\lambda_{ii} > 0$. Hence we have
\begin{equation}\label{eq41}
  {\mathbf{R}_X} = {{\mathbf{W}}_1} + {{\mathbf{W}}_2} = \sum\nolimits_{i = 1}^{{N_t}} {{\lambda _{ii}}} {{\mathbf{u}}_i}{\mathbf{u}}_i^H = {\mathbf{U\Lambda }}{{\mathbf{U}}^H},
\end{equation}
where $\lambda_{ii} = {{\mathbf{u}}_i^H{{\mathbf{R}}_X}{{\mathbf{u}}_i}}$, and ${\mathbf{\Lambda }} = \operatorname{diag} \left( {{\lambda _{11}},{\lambda _{22}}, \ldots {\lambda _{{N_t}{N_t}}}} \right)$. This completes the proof.
\end{proof}
With Lemma 2 at hand, the optimal solution can be attained by the following theorem.
\begin{thm}
  The optimal solution of problem (\ref{eq38}) can be given as
  \begin{equation}\label{eq42}
   {{\mathbf{W}}_1} = \frac{{{P_T}}}{{{N_t}}}\frac{{{\mathbf{h}}_1{\mathbf{h}}_1^H}}{{{{\left\| {{{\mathbf{h}}_1}} \right\|}^2}}},\quad {{\mathbf{R}}_X} = \frac{{{P_T}}}{{{N_t}}}{{\mathbf{I}}_{{N_t}}},
  \end{equation}
  if ${\Gamma _1} < \frac{{{P_T}{{\left\| {{{\mathbf{h}}_{\mathbf{1}}}} \right\|}^2}}}{{{N_t}\sigma _C^2}}$, and
  \begin{equation}\label{eq43}
    {{\mathbf{W}}_1} = \frac{{{\Gamma _1}\sigma _C^2{\mathbf{h}}_1{\mathbf{h}}_1^H}}{{{{\left\| {{{\mathbf{h}}_1}} \right\|}^4}}},{{\mathbf{R}}_X} = \sum\nolimits_{i = 1}^{{N_t}} {{\lambda _{ii}}{{\mathbf{u}}_i}} {\mathbf{u}}_i^H,
  \end{equation}
  if $\frac{{{P_T}{{\left\| {{{\mathbf{h}}_{\mathbf{1}}}} \right\|}^2}}}{{{N_t}\sigma _C^2}} \le {\Gamma _1} \le \frac{{{P_T}{{\left\| {{{\mathbf{h}}_{\mathbf{1}}}} \right\|}^2}}}{{\sigma _C^2}}$, where $\lambda_{ii},\forall i$ are computed as
  \begin{equation}\label{eq44}
    {\lambda_{11}} = \frac{{{\Gamma _1}\sigma _C^2}}{{{{\left\| {{{\mathbf{h}}_1}} \right\|}^2}}}, {\lambda _{ii}} = \frac{{{P_T}{{\left\| {{{\mathbf{h}}_1}} \right\|}^2} - {\Gamma _1}\sigma _C^2}}{{{{\left\| {{{\mathbf{h}}_1}} \right\|}^2}\left( {{N_t} - 1} \right)}},i = 2,3, \ldots ,{N_t}.
  \end{equation}
\end{thm}
\begin{proof}
See Appendix E.
\end{proof}
The closed-form solutions obtained from Theorem 3 naturally satisfy the rank constraints, i.e., $\operatorname{rank}\left(\mathbf{W}_1\right) = 1$. Hence it is also the optimal solution to (\ref{eq35}) for the single-user case.
\subsection{Rank-One Optimal Solution of (\ref{eq37}) in the Multi-User Case}
Although the convex relaxation (\ref{eq37}) can be optimally solved using numerical tools, it is not guaranteed to yield rank-one or low-rank solutions. In the case that the solutions are with high ranks, one may obtain the low-rank approximations by employing various of methods, e.g., eigenvalue decomposition or Gaussian randomization. Nevertheless, the eigenvalue based approximation might not be accurate given the matrix inverse operation involved in the objective function, and the Gaussian randomization could be computationally expensive. Below we propose a constructive method to extract the exact rank-one optimum directly from the solutions of the convex-relaxation problem (\ref{eq37}).
\\\indent Similar to (\ref{eq38}), problem (\ref{eq37}) can be equivalently formulated as
\begin{equation}\label{eq45}
  \begin{gathered}
  \mathop {\min }\limits_{\left\{ {{{\mathbf{W}}_k}} \right\}_{k = 1}^K,{{\mathbf{R}}_X}} \;\operatorname{tr} \left( {{\mathbf{R}}_X^{ - 1}} \right) \hfill \\
  \;\;\;\;\;\;\text{s.t.}\;\operatorname{tr} \left( {{{\mathbf{Q}}_k}{{\mathbf{W}}_k}} \right) - {\Gamma _k}\operatorname{tr} \left( {{{\mathbf{Q}}_k}\left( {{{\mathbf{R}}_X} - {{\mathbf{W}}_k}} \right)} \right) \ge {\Gamma _k}\sigma _C^2,\forall k, \hfill \\
  \;\;\;\;\;\;\;\;\;\;\;\;\operatorname{tr} \left( {{{\mathbf{R}}_X}} \right) \le {P_T},{{\mathbf{R}}_X}  \succeq \sum\nolimits_{k = 1}^K {{{\mathbf{W}}_k}} ,{{\mathbf{W}}_k} \succeq {\mathbf{0}},\forall k. \hfill \\
\end{gathered}
\end{equation}
By denoting the optimal solution of (\ref{eq45}) as ${{{\mathbf{\bar R}}}_X},\left\{ {{{{\mathbf{\bar W}}}_k}} \right\}_{k = 1}^K$, we have $\operatorname{rank}\left({{{\mathbf{\bar R}}}_X}\right) = N_t, \operatorname{rank}\left({{{\mathbf{\bar W}}}_k}\right) \ge 1, \forall k$. If $\operatorname{rank}\left({{{\mathbf{\bar W}}}_k}\right) = 1, \forall k$, then ${{{\mathbf{\bar R}}}_X},\left\{ {{{{\mathbf{\bar W}}}_k}} \right\}_{k = 1}^K$ are also optimal for (\ref{eq35}). Otherwise, one can extract the optimal solution of (\ref{eq35}) from ${{{\mathbf{\bar R}}}_X},\left\{ {{{{\mathbf{\bar W}}}_k}} \right\}_{k = 1}^K$ by relying on the following theorem.
\begin{thm}
  Given an optimal solution ${{{\mathbf{\bar R}}}_X},\left\{ {{{{\mathbf{\bar W}}}_k}} \right\}_{k = 1}^K$ of (\ref{eq45}), the following ${{{\mathbf{\tilde R}}}_X},\left\{ {{{{\mathbf{\tilde W}}}_k}} \right\}_{k = 1}^K$ is also an optimal solution:
  \begin{equation}\label{eq46}
    {{{\mathbf{\tilde R}}}_X} = {{{\mathbf{\bar R}}}_X},\quad {{{\mathbf{\tilde W}}}_k} = \frac{{{{{\mathbf{\bar W}}}_k}{{\mathbf{Q}}_k}{\mathbf{\bar W}}_k^H}}{{\operatorname{tr} \left( {{{\mathbf{Q}}_k}{{{\mathbf{\bar W}}}_k}} \right)}}, \forall k \le K,
  \end{equation}
  where $\operatorname{rank}\left({{{\mathbf{\tilde W}}}_k}\right) = 1, \forall k \le K$.
\end{thm}
\begin{proof}
  {\color{red}See \cite{9124713}.}
\end{proof}
The idea behind Theorem 4 is simple, i.e., to preserve the useful signal power of the $k$-th user by replacing ${{{\mathbf{\bar W}}}_k}$ with the rank-one matrix ${{{\mathbf{\tilde W}}}_k}$, and then put the difference between ${{{\mathbf{\tilde W}}}_k}$ and ${{{\mathbf{\bar W}}}_k}$ into the covariance matrix of the extra radar beamformer. This guarantees that $\mathbf{\bar R}_X$ is unchanged, and thus the objective value and the transmit power remain the same. Moreover, the useful signal power and the interference keep unchanged for each user, which suggests that the resulting SINR is the same as before. Therefore, ${{{\mathbf{\tilde R}}}_X},\left\{ {{{{\mathbf{\tilde W}}}_k}} \right\}_{k = 1}^K$ is feasible and optimal. We refer readers to \cite{9124713} for a detailed proof. Based on Theorem 4, the first $K$ columns (the communication beamformer $\mathbf{W}_C$) of the optimal beamformer $\mathbf{\tilde W}_D$ for the original problem (\ref{eq35}) can be straightforwardly expressed as
\begin{equation}\label{eq47}
  {{\mathbf{w}}_k} = \frac{{{{{\mathbf{\bar W}}}_k}{{\mathbf{h}}_k}}}{{\sqrt {\operatorname{tr} \left( {{{\mathbf{Q}}_k}{{{\mathbf{\bar W}}}_k}} \right)} }} = {\left( {{\mathbf{h}}_k^H{{{\mathbf{\bar W}}}_k}{{\mathbf{h}}_k}} \right)^{ - 1/2}}{{{\mathbf{\bar W}}}_k}{{\mathbf{h}}_k}, \forall k \le K.
\end{equation}
Accordingly, the auxiliary beamformer $\mathbf{W}_A$ can be attained as a square-root of ${{\mathbf{\tilde R}}_X} - \sum\nolimits_{k = 1}^K {{{\mathbf{\tilde W}}_k}}$, i.e.,
\begin{equation}\label{eq48}
  {{\mathbf{W}}_A}{\mathbf{W}}_A^H = {{{\mathbf{\tilde R}}}_X} - \sum\nolimits_{k = 1}^K {{{{\mathbf{\tilde W}}}_k}},
\end{equation}
where various approaches can be exploited to extract ${\mathbf{W}}_A$, e.g., Cholesky decomposition or eigenvalue decomposition.

\section{Numerical Results}
In this section, we provide numerical results to verify the advantage of the proposed joint beamforming approaches. Without loss of generality, we consider a DFRC BS that is equipped with $N_t = 16$ and $N_r = 20$ antennas for its transmitter and receiver. The power budget is $P_T = 30\text{dBm}$, the noise variances are set as $\sigma_C^2 = \sigma_R^2 = 0\text{dBm}$, and the DFRC frame length is set as $L = 30$. Moreover, for point target scenario, we assume that the target angle is $\theta = 0^\circ$. In the case of extended targets, we assume that the entries of the target response matrix $\mathbf{G}$ are i.i.d. Gaussian distributed with zero mean and unit variance. Since for extended target estimation, the CRB is equal to the MSE, we will use MSE to measure the estimation performance of $\mathbf{G}$.
\subsection{Verification of the Closed-form Solutions}
We commence by examining the correctness of the closed-form solutions attained for both point and extended target scenarios, with the presence of a single communication user. The results are shown in Fig. 2, where the radar estimation performance for point and extended targets are shown via the root-CRB of the target angle and the MSE of the target response matrix, respectively. The closed-form solutions match well with their numerical counterparts. Moreover, the increase of the required SINR at the users leads to rising CRB and MSE. Fortunately, the target estimation errors can be maintained at the lowest level for both cases when the required SINR is below 30dB.
\begin{figure}[!t]
    \centering
    \includegraphics[width=0.95\columnwidth]{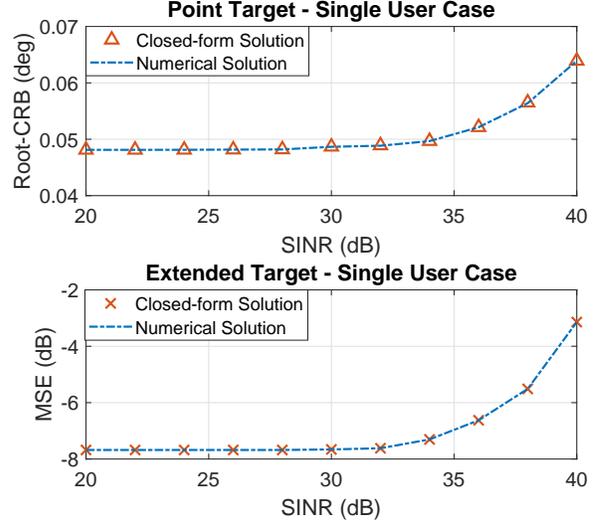}
    \caption{Closed-form and numerical solutions for the single-user scenario.}
    \label{fig:2}
\end{figure}
\subsection{Joint Beamforming for Point Target and Multiple Users}
Next, we investigate the performance of the proposed joint beamforming method for the scenario of point target and multiple users in Figs. 3--5. Our benchmark techniques are the DFRC beamforming schemes proposed in \cite{9124713} and \cite{8288677}, where the beamforming matrix is designed such that a given radar-only beampattern is achieved/approximated under individual SINR constraints for downlink communication users. For convenience, we use the term ``Beampattern Approx. Design 1", ``Beampattern Approx. Design 2", and ``Proposed CRB-Min Design" to represent the DFRC beamforming technique in \cite{9124713}, \cite{8288677}, and the proposed method, respectively. For the beampattern approximation methods, we define a 3dB beamwidth of $10^\circ$. We first show the resultant beampatterns of the three techniques in Fig. 3, where the number of users and their required SINRs are assumed to be $K = 4$ and $\Gamma_k = 15\text{dB}, \forall k$, respectively. It can be seen that all the three beamformers correctly focus their mainlobe towards $0^\circ$. We observe that all the obtained beampatterns show random fluctuations in their sidelobe regions, due to the imposed SINR constraints for users. Moreover, the proposed CRB-Min method radiates the highest power towards the target angle among all the three techniques.
\\\indent The performance for target estimation is explicitly shown in Fig. 4 in terms of the root-MSE (RMSE), with the increase of the receive SNR of the echo signal, which is defined as $\operatorname{SNR}_{\text{radar}} = \frac{\left|\alpha\right|^2LP_T}{\sigma_R^2}$. We obtain the MLE for the target angle via exhaustive search \cite{1703855}. As expected, the RMSE is lower-bounded by the corresponding CRB, and in particular, the CRB is tight and can be achieved by the MLE in the high-SNR regime. It can be seen that the proposed method outperforms both beampattern approximation designs, especially when the SNR is low. This proves that the proposed CRB-Min technique can indeed improve the target estimation performance, as compared to conventional approaches.
\begin{figure}[!t]
    \centering
    \includegraphics[width=0.95\columnwidth]{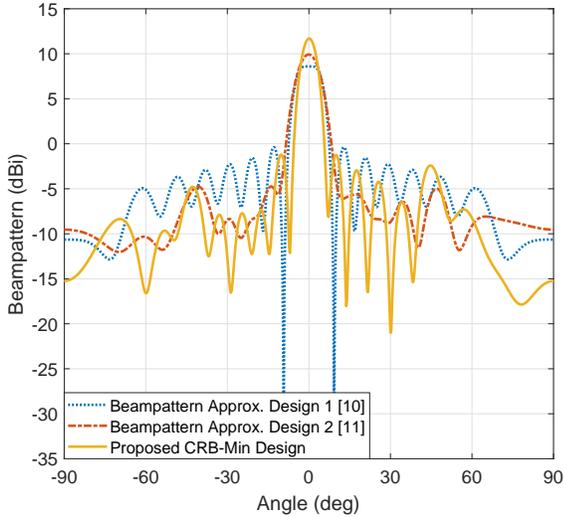}
    \caption{Beampatterns for the scenario of point target and multiple users, with the method proposed in \cite{9124713} and \cite{8288677} as benchmarks. The number of users is $K = 4$, and the SINR threshold is set as $15\text{dB}$.}
    \label{fig:3}
\end{figure}
\begin{figure}[!t]
    \centering
    \includegraphics[width=0.95\columnwidth]{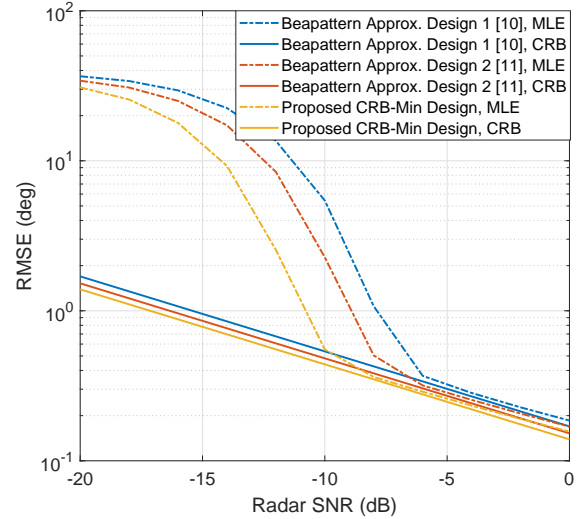}
    \caption{Target estimation performance in the scenario of point target and multiple users, with the method proposed in \cite{9124713} and \cite{8288677} as benchmarks. The number of users is $K = 4$, and the SINR threshold is set as $15\text{dB}$.}
    \label{fig:4}
\end{figure}
\\\indent In Fig. 5, we consider the performance tradeoff between radar and communication. Owing to the increasing SINR threshold of the users, the CRB for target estimation becomes higher. For a smaller number of users, however, the CRB remains at a low level despite that the user's SINR is growing. Again, we see that the proposed technique is superior to both beampattern approximation methods in \cite{9124713} and \cite{8288677}.
\begin{figure}[!t]
    \centering
    \includegraphics[width=0.95\columnwidth]{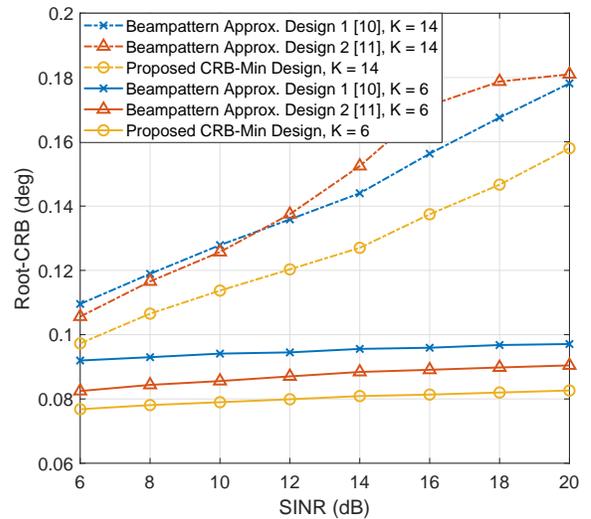}
    \caption{Tradeoff between radar and communication performance in the scenario of point target and multiple users, with the method proposed in \cite{9124713} and \cite{8288677} as benchmarks. The number of users is set as $K = 6$ and $K = 12$, respectively.}
    \label{fig:5}
\end{figure}
\begin{figure}[!t]
    \centering
    \includegraphics[width=0.95\columnwidth]{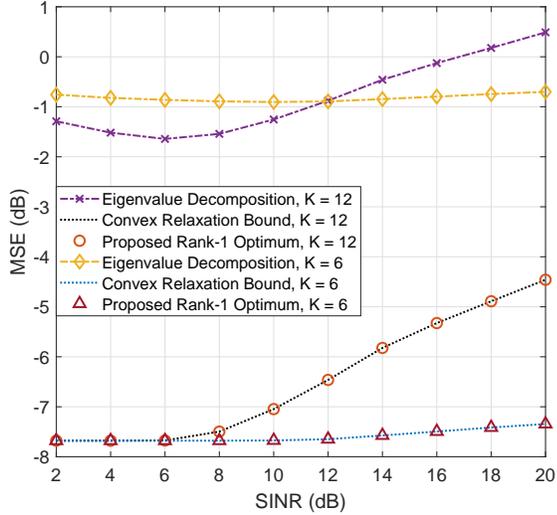}
    \caption{Tradeoff between target estimation MSE and users' SINR for the scenario of extended target and multiple users, in the cases of $K = 12$ and $K = 6$. The eigenvalue decomposition based rank-one approximation of (\ref{eq37}) serves the benchmark.}
    \label{fig:6}
\end{figure}
\begin{figure}[!t]
    \centering
    \includegraphics[width=0.95\columnwidth]{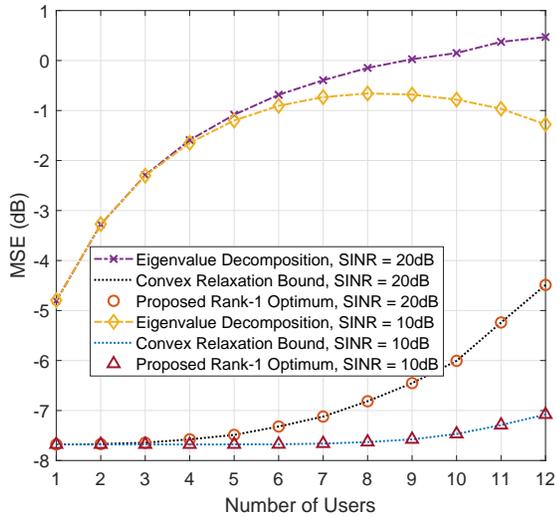}
    \caption{Target estimation MSE with an increasing user number for the scenario of extended target and multiple users, with SINR of $20\text{dB}$ and $10\text{dB}$. The eigenvalue decomposition based rank-one approximation of (\ref{eq37}) serves as the benchmark.}
    \label{fig:7}
\end{figure}
\subsection{Joint Beamforming for Extended Target and Multiple Users}
We study the scenario of the extended target and multiple users in Figs. 6--7. In Fig. 6, we plot the performance tradeoff between the target estimation MSE and the required SINR for users with $K = 12$ and $6$, respectively. The rank-one approximation of the SDR problem (\ref{eq37}) is employed as a benchmark, which is obtained by applying the eigenvalue decomposition on the solution of (\ref{eq37}). We see that by exploiting Theorem 4, we can indeed acquire the globally optimal solution, which is superior to the conventional rank-one approximation methods, e.g., eigenvalue decomposition. On the other hand, the eigenvalue decomposition fails to generate a favorable performance tradeoff between MSE and SINR, as the trends of the corresponding tradeoff curves are not monotonically increasing. More interestingly, the MSE for the proposed rank-one optimum remains at a small level despite that the SINR requirement is increasing, given a moderate number of users, e.g., $K = 6$.
\\\indent Finally in Fig. 7, we show the impact of the communication user number imposed on the radar estimation performance, with the SINR being set as $20\text{dB}$ and $10\text{dB}$, respectively. It is observed that given an increasing user number, the estimation performance becomes worse. Fortunately, the variation of the MSE can be kept within 1dB when the required SINR is 10dB. The results prove again the performance gain of the proposed rank-one optimum over that of the eigenvalue decomposition.
\section{Conclusion}
In this paper, we proposed beamforming designs for joint radar sensing and multi-user communications, for the scenarios of point and extended targets, respectively. In particular, we formulated optimization problems to minimize the CRB of target estimation by imposing SINR constraints for multiple communication users. While the considered problems are non-convex, we derived closed-form optimal solutions for both cases in the presence of a single user. In the scenario of multiple users, we designed the DFRC beamformers by exploiting the SDR approach. We then proved that the globally optimal solutions are achievable for both problems. Numerical results demonstrate that the proposed approaches reach the globally optimal solutions, while significantly outperforming the benchmark techniques in terms of target estimation performance.

\appendices

\section{Proof of Lemma 1}
Consider an optimal solution ${{\mathbf{w}}_1}$. We can always express ${{\mathbf{w}}_1}$ as
  \begin{equation}\label{eq49}
    {{\mathbf{w}}_1} = a{{\mathbf{u}}_\alpha } + b{{\mathbf{u}}_\beta },
  \end{equation}
where $\left\| {{{\mathbf{u}}_\alpha }} \right\| = \left\| {{{\mathbf{u}}_\beta }} \right\| = 1$, and ${{\mathbf{u}}_\alpha } \in \operatorname{span}\left\{ {{\mathbf{a}},{{\mathbf{h}}_1}} \right\}$, ${{\mathbf{u}}_\beta } \bot \operatorname{span}\left\{ {{\mathbf{a}},{{\mathbf{h}}_1}}\right\}$, which are the normalized projections of ${{\mathbf{w}}_1}$ onto $\operatorname{span}\left\{ {{\mathbf{a}},{{\mathbf{h}}_1}} \right\}$ and its null space, respectively. Since the SINR constraint and the power budget are satisfied, we have
\begin{equation}\label{eq50}
{a^2}{\left| {{\mathbf{h}}_1^H{{\mathbf{u}}_\alpha }} \right|^2} \ge \Gamma_1 \sigma _C^2,\quad {\left\| {{{\mathbf{w}}_1}} \right\|^2} = {a^2} + {b^2} \le {P_T},
\end{equation}
due to the fact that only ${{\mathbf{u}}_\alpha }$ contributes to the SINR. Therefore, one can always let $a = \sqrt{P_T}, b = 0$, in which case the objective function is strictly increased without violating the constraints. This implies that the optimal $\mathbf{w}_1$ belongs to $\operatorname{span}\left\{ {{\mathbf{a}},{{\mathbf{h}}_1}} \right\}$.

\section{Proof of Theorem 1}
We first show that the power budget of $\mathbf{w}_1$ should always be fully exploited in order to maximize the objective function ${\left| {{{\mathbf{a}}^H}{{\mathbf{w}}_1}} \right|^2}$. Suppose that there is an optimal solution ${{\mathbf{\tilde w}}_1}$, such that ${\left\| {{\mathbf{\tilde w}}_1} \right\|^2} = \tilde P < {P_T}$. Then we have
\begin{equation}\label{eq51}
  {\left| {{\mathbf{h}}_1^H{{{\mathbf{\tilde w}}}_1}} \right|^2} \ge {\Gamma _1}\sigma _C^2.
\end{equation}
Now let us consider another solution ${{{\mathbf{\bar w}}}_1} = \sqrt {\frac{{{P_T}}}{{\tilde P}}} {{{\mathbf{\tilde w}}}_1}$, which has the power budget of $P_T$. It can be readily verified that
\begin{equation}\label{eq52}
\begin{gathered}
  {\left| {{\mathbf{h}}_1^H{{{\mathbf{\bar w}}}_1}} \right|^2} = \frac{{{P_T}}}{{\tilde P}}{\left| {{\mathbf{h}}_1^H{{{\mathbf{\tilde w}}}_1}} \right|^2} > {\Gamma _1}\sigma _C^2, \hfill \\
  {\left| {{{\mathbf{a}}^H}{{{\mathbf{\bar w}}}_1}} \right|^2} = \frac{{{P_T}}}{{\tilde P}}{\left| {{{\mathbf{a}}^H}{{{\mathbf{\tilde w}}}_1}} \right|^2} > {\left| {{{\mathbf{a}}^H}{{{\mathbf{\tilde w}}}_1}} \right|^2}, \hfill \\
\end{gathered}
\end{equation}
which implies that ${{{\mathbf{\bar w}}}_1}$ is feasible solution that generates higher objective value than that of ${{\mathbf{\tilde w}}_1}$. Therefore, the power budget is fully exploited when the optimum is reached.
\\\indent By noting the above fact, we next consider the case where the SINR constraint is inactive. The solution can be readily obtained by fully allocating the power along the direction of $\mathbf{a}$, i.e., ${{\mathbf{w}}_1} = \sqrt {{P_T}} \frac{{\mathbf{a}}}{{\left\| {\mathbf{a}} \right\|}}$. Now let us discuss the case that the SINR constraint is active. By noting Lemma 1, the optimal $\mathbf{w}_1$ can be expressed as
\begin{equation}\label{eq53}
  \mathbf{w}_1 = x_1{{{\mathbf{u}}}_1}+ x_2{\mathbf{a}}_u,\quad x_1, x_2 \in \mathbb{C},
\end{equation}
since $\operatorname{span}\left\{{\mathbf{a}}_u,{{{\mathbf{u}}}_1}\right\} = \operatorname{span}\left\{{\mathbf{a}},{{{\mathbf{h}}}_1}\right\}$. Accordingly, the problem can be reformulated as
\begin{equation}\label{eq54}
\begin{gathered}
  \mathop {\max }\limits_{{x_1},{x_2}} \;{\left| {{x_1}{{\mathbf{a}}^H}{{\mathbf{u}}_1}+{x_2}{{\mathbf{a}}^H}{{\mathbf{a}}_u}} \right|^2} \hfill \\
  \text{\text{s.t.}}\;\;\;\;{\left| {{x_1}} \right|^2}{\left\| {{{\mathbf{h}}_1}} \right\|^2} = \Gamma_1 \sigma _C^2, \hfill \\
  \;\;\;\;\;\;\;\;\;{\left| {{x_1}} \right|^2}+{\left| {{x_2}} \right|^2} = {P_T}. \hfill \\
\end{gathered}
\end{equation}
It follows that
\begin{equation}\label{eq55}
  {\left| {{x_1}} \right|^2} = \frac{{\Gamma_1 \sigma _C^2}}{{{{\left\| {{{\mathbf{h}}_1}} \right\|}^2}}},{\left| {{x_2}} \right|^2} = {P_T} - \frac{{\Gamma_1 \sigma _C^2}}{{{{\left\| {{{\mathbf{h}}_1}} \right\|}^2}}}.
\end{equation}
To maximize the objective function, the phases of $x_1$ and $x_2$ should be the opposite of that of ${{\mathbf{a}}^H}{{\mathbf{u}}_1}$ and ${{\mathbf{a}}^H}{\mathbf{a}}_u$, respectively, i.e., $x_1$ and $x_2$ should be aligned with the directions of ${{\mathbf{a}}^H}{{\mathbf{u}}_1}$ and ${{\mathbf{a}}^H}{\mathbf{a}}_u$. This results in the expressions in (\ref{eq29}), which completes the proof.

\section{Proof of Proposition 1}
Minimizing $\operatorname{CRB} \left( \theta  \right)$  is equivalent to
  \begin{equation}\label{eq56}
  \begin{gathered}
  \mathop {\max }\limits_{{{\mathbf{R}}_X}\succeq \mathbf{0},t} \;t \hfill \\
  \;\;\text{\text{s.t.}}\;\;\operatorname{tr} \left( {{{{\mathbf{\dot A}}}^H}{\mathbf{\dot A}}{{\mathbf{R}}_X}} \right) \hfill \\
  \;\;\;\;\;\; - {\left| {\operatorname{tr} \left( {{{{\mathbf{\dot A}}}^H}{\mathbf{A}}{{\mathbf{R}}_X}} \right)} \right|^2}{\operatorname{tr} ^{ - 1}}\left( {{{\mathbf{A}}^H}{\mathbf{A}}{{\mathbf{R}}_X}} \right) \ge t. \hfill \\
  \end{gathered}
  \end{equation}
By leveraging the following Schur complement condition \cite{zhang2006schur}
  \begin{equation}\label{eq57}
  \begin{gathered}
  \operatorname{tr} \left( {{{{\mathbf{\dot A}}}^H}{\mathbf{\dot A}}{{\mathbf{R}}_X}} \right) - t - {\left| {\operatorname{tr} \left( {{{{\mathbf{\dot A}}}^H}{\mathbf{A}}{{\mathbf{R}}_X}} \right)} \right|^2}{\operatorname{tr} ^{ - 1}}\left( {{{\mathbf{A}}^H}{\mathbf{A}}{{\mathbf{R}}_X}} \right) \ge 0 \hfill \\
   \Leftrightarrow \left[ {\begin{array}{*{20}{c}}
  {\operatorname{tr} \left( {{{{\mathbf{\dot A}}}^H}{\mathbf{\dot A}}{{\mathbf{R}}_X}} \right) - t}&{\operatorname{tr} \left( {{{{\mathbf{\dot A}}}^H}{\mathbf{A}}{{\mathbf{R}}_X}} \right)} \\
  {\operatorname{tr} \left( {{{\mathbf{A}}^H}{\mathbf{\dot A}}{{\mathbf{R}}_X}} \right)}&{\operatorname{tr} \left( {{{\mathbf{A}}^H}{\mathbf{A}}{{\mathbf{R}}_X}} \right)}
  \end{array}} \right] \succeq {\mathbf{0}}, \hfill \\
  \end{gathered}
  \end{equation}
problem (\ref{eq56}) can be rewritten as (\ref{eq30}), which completes the proof.
\section{Proof of Theorem 2}
Let us define the dual variables for problem (\ref{eq34}), which are $\left\{ {{\mu_1},{\mu_2}, \ldots, {\mu_{K }},{\mu_{T}}} \right\}$ that are associated with $K+1$ linear constraints, and $\left\{ {{\mathbf{Z}_1},{\mathbf{Z}_2}, \ldots,{\mathbf{Z}_{K}},{\mathbf{Z}_{P}}} \right\} \succeq \mathbf{0}$ that are associated with $K+1$ semidefinite constraints. By assuming that the optimality is reached with $\left\{ {{\mu_1},{\mu_2}, \ldots, {\mu_{K }},{\mu_{T}}} \right\}$, $\left\{ {{\mathbf{Z}_1},{\mathbf{Z}_2}, \ldots,{\mathbf{Z}_{K}}, {\mathbf{Z}_{P}}} \right\}$ and $\left\{ {{\mathbf{W}_1},{\mathbf{W}_2}, \ldots, {\mathbf{W}_{K}}} \right\}$, the following complementary conditions hold true
\begin{equation}\label{eq58}
  \begin{gathered}
  -{\mu_k}\left( {\operatorname{tr} \left( {{{\mathbf{Q}}_k}{{\mathbf{W}}_k}} \right) - {\Gamma _k}\sum\nolimits_{i = 1,i \ne k}^{{K}} {\operatorname{tr} \left( {{{\mathbf{Q}}_k}{{\mathbf{W}}_i}} \right) - {\Gamma _k}\sigma _C^2} } \right) = 0,\hfill \\
  {\mu_k} \ge 0,\forall k,\hfill \\
  {\mu_{T}}\left( {\sum\limits_{k = 1}^{{K}} {\operatorname{tr} \left( {{{\mathbf{W}}_k}} \right)}  - {P_T}} \right) = 0,{\mu_{T}} \ge 0, \hfill \\
  \operatorname{tr}\left({\mathbf{Z}}_k{{\mathbf{W}}_k}\right) = 0,{{\mathbf{Z}}_k} \succeq {\mathbf{0}}, \forall k, \hfill \\
  \operatorname{tr}\left({\mathbf{Z}}_{P}\mathbf{P}\right)= 0,{{\mathbf{Z}}_{P}} \succeq {\mathbf{0}},\hfill \\
\end{gathered}
\end{equation}
where
\begin{equation}\label{eq59}
\begin{gathered}
  {\mathbf{P}} \triangleq \left[ {\begin{array}{*{20}{c}}
  {\operatorname{tr} \left( {{{{\mathbf{\dot A}}}^H}{\mathbf{\dot A}}{{\mathbf{R}}_X}} \right) - t}&{\operatorname{tr} \left( {{{{\mathbf{\dot A}}}^H}{\mathbf{A}}{{\mathbf{R}}_X}} \right)} \\
  {\operatorname{tr} \left( {{{\mathbf{A}}^H}{\mathbf{\dot A}}{{\mathbf{R}}_X}} \right)}&{\operatorname{tr} \left( {{{\mathbf{A}}^H}{\mathbf{A}}{{\mathbf{R}}_X}} \right)}
\end{array}} \right] \hfill \\
   = \left[ {\begin{array}{*{20}{c}}
  {{{\left\| {{\mathbf{\dot b}}} \right\|}^2}{{\mathbf{a}}^H}{{\mathbf{R}}_X}{\mathbf{a}} + {{\left\| {\mathbf{b}} \right\|}^2}{{{\mathbf{\dot a}}}^H}{{\mathbf{R}}_X}{\mathbf{\dot a}} - t}&{{{\left\| {\mathbf{b}} \right\|}^2}{{\mathbf{a}}^H}{{\mathbf{R}}_X}{\mathbf{\dot a}}} \\
  {{{\left\| {\mathbf{b}} \right\|}^2}{{{\mathbf{\dot a}}}^H}{{\mathbf{R}}_X}{\mathbf{a}}}&{{{\left\| {\mathbf{b}} \right\|}^2}{{\mathbf{a}}^H}{{\mathbf{R}}_X}{\mathbf{a}}}
\end{array}} \right], \hfill \\
\end{gathered}
\end{equation}
where the second equality holds due to the orthogonal property in (\ref{eq22}). Moreover, the Lagrangian can be formulated as
\begin{equation}\label{eq60}
\begin{gathered}
  \mathcal{L} =  - t - \operatorname{tr} \left( {{{\mathbf{Z}}_{P}}{\mathbf{P}}} \right) - \sum\nolimits_{k = 1}^K {\operatorname{tr} \left( {{{\mathbf{Z}}_k}{{\mathbf{W}}_k}} \right)}-  \hfill \\
   \sum\nolimits_{k = 1}^K {{\mu _k}\left( {\operatorname{tr} \left( {{{\mathbf{Q}}_k}{{\mathbf{W}}_k}} \right) - {\Gamma _k}\sum\nolimits_{i = 1,i \ne k}^K {\operatorname{tr} \left( {{{\mathbf{Q}}_k}{{\mathbf{W}}_i}} \right) - {\Gamma _k}\sigma _C^2} } \right)}  \hfill \\
   + {\mu _{T}}\left( {\sum\limits_{k = 1}^K {\operatorname{tr} \left( {{{\mathbf{W}}_k}} \right)}  - {P_T}} \right). \hfill \\
\end{gathered}
\end{equation}
Let
\begin{equation}\label{eq61}
  {{\mathbf{Z}}_{P}} = \left[ {\begin{array}{*{20}{c}}
  \phi &\beta  \\
  {{\beta ^*}}&\gamma
\end{array}} \right] \succeq \mathbf{0}.
\end{equation}
The derivative of the Lagrangian at the optimum can be given as
\begin{equation}\label{eq62}
\begin{gathered}
  \frac{{\partial \mathcal{L}}}{{\partial t}} =  - 1 + \phi  = 0 \Leftrightarrow \phi  = 1, \hfill \\
  \frac{{\partial \mathcal{L}}}{{\partial {{\mathbf{W}}_k}}} =  - {\mathbf{F}} - {{\mathbf{Z}}_k} - {\mu _k}\left( {1 + {\Gamma _k}} \right){{\mathbf{Q}}_k} \hfill \\
  \;\;\;\;\;\;\;\;\;\;\; + \sum\nolimits_{i = 1}^K {{\mu _i}{\Gamma _i}{{\mathbf{Q}}_i} + } {\mu _{T}}{{\mathbf{I}}_{{N_t}}} = \mathbf{0},\forall k, \hfill \\
\end{gathered}
\end{equation}
where
\begin{equation}\label{eq63}
\begin{gathered}
  {\mathbf{F}} \triangleq \frac{{\partial \operatorname{tr} \left( {{{\mathbf{Z}}_{P}}{\mathbf{P}}} \right)}}{{\partial {\mathbf{W}}_k}} \hfill \\
   = \left( {\phi {{\left\| {{\mathbf{\dot b}}} \right\|}^2} + \gamma {{\left\| {\mathbf{b}} \right\|}^2}} \right){\mathbf{a}}{{\mathbf{a}}^H} + \phi {\left\| {\mathbf{b}} \right\|^2}{\mathbf{\dot a}}{{{\mathbf{\dot a}}}^H} + 2{\left\| {\mathbf{b}} \right\|^2}\operatorname{Re} \left( {\beta {\mathbf{a}}{{{\mathbf{\dot a}}}^H}} \right) \hfill \\
   = \left[ {\begin{array}{*{20}{c}}
  {\mathbf{a}}&{{\mathbf{\dot a}}}
\end{array}} \right]\left[ {\begin{array}{*{20}{c}}
  {\phi {{\left\| {{\mathbf{\dot b}}} \right\|}^2} + \gamma {{\left\| {\mathbf{b}} \right\|}^2}}&{\beta {{\left\| {\mathbf{b}} \right\|}^2}} \\
  {{\beta ^*}{{\left\| {\mathbf{b}} \right\|}^2}}&{\phi {{\left\| {\mathbf{b}} \right\|}^2}}
\end{array}} \right]\left[ \begin{gathered}
  {{\mathbf{a}}^H} \hfill \\
  {{{\mathbf{\dot a}}}^H} \hfill \\
\end{gathered}  \right]. \hfill \\
\end{gathered}
\end{equation}
Since $\phi = 1$, it follows that ${{\mathbf{Z}}_{P}} \ne \mathbf{0}$. By noting that $\mathbf{P} \ne  \mathbf{0}$, both ${{\mathbf{Z}}_{P}}$ and $\mathbf{P}$ should be singular matrices to satisfy ${\operatorname{tr} \left( {{{\mathbf{Z}}_{P}}{\mathbf{P}}} \right)} = 0$. For ${{\mathbf{Z}}_{P}} \succeq \mathbf{0}$, this implies
\begin{equation}\label{eq64}
  \phi  - {\left| \beta  \right|^2}{\gamma ^{ - 1}} = 1 - {\left| \beta  \right|^2}{\gamma ^{ - 1}} = 0 \Leftrightarrow \gamma  = {\left| \beta  \right|^2}.
\end{equation}
By leveraging the relationship in (\ref{eq64}), it can be verified that
\begin{equation}\label{eq65}
\begin{gathered}
  \left[ {\begin{array}{*{20}{c}}
  {\phi {{\left\| {{\mathbf{\dot b}}} \right\|}^2} + \gamma {{\left\| {\mathbf{b}} \right\|}^2}}&{\beta {{\left\| {\mathbf{b}} \right\|}^2}} \\
  {{\beta ^*}{{\left\| {\mathbf{b}} \right\|}^2}}&{\phi {{\left\| {\mathbf{b}} \right\|}^2}}
\end{array}} \right] \hfill \\
   = \left[ {\begin{array}{*{20}{c}}
  {{{\left\| {{\mathbf{\dot b}}} \right\|}^2} + {{\left| \beta  \right|}^2}{{\left\| {\mathbf{b}} \right\|}^2}}&{\beta {{\left\| {\mathbf{b}} \right\|}^2}} \\
  {{\beta ^*}{{\left\| {\mathbf{b}} \right\|}^2}}&{{{\left\| {\mathbf{b}} \right\|}^2}}
\end{array}} \right] \succ {\mathbf{0}}. \hfill \\
\end{gathered}
\end{equation}
Moreover, since ${\mathbf{a}} \bot {\mathbf{\dot a}}$, we have $\mathbf{F} \succeq \mathbf{0}, \operatorname{rank}\left(\mathbf{F}\right) = 2$. The two non-zero eigenvalues of $\mathbf{F}$, denoted as $\lambda_1$ and $\lambda_2$, are shown in (\ref{eq66}). It can be readily observed that $\lambda_1 = \lambda_{\max}\left(\mathbf{F}\right) > \lambda_2$ regardless of the value of $\beta$, under the condition that $N_t \ne N_r$, which is satisfied in general for MIMO radar where $N_r$ is chosen to be larger than $N_t$.
\begin{figure*}[ht]
\normalsize
\newcounter{MYtempeqncnt3}
\setcounter{MYtempeqncnt3}{\value{equation}}
\setcounter{equation}{64}
\begin{equation}\label{eq66}
\begin{gathered}
  {\lambda _1} = \frac{{{N_t}\left( {{{\left\| {{\mathbf{\dot b}}} \right\|}^2} + {{\left| \beta  \right|}^2}{N_r}} \right) + {N_r}{{\left\| {{\mathbf{\dot a}}} \right\|}^2} + \sqrt {{{\left( {{N_t}\left( {{{\left\| {{\mathbf{\dot b}}} \right\|}^2} + {{\left| \beta  \right|}^2}{N_r}} \right) - {N_r}{{\left\| {{\mathbf{\dot a}}} \right\|}^2}} \right)}^2} + 4{{\left| \beta  \right|}^2}{N_t}N_r^2{{\left\| {{\mathbf{\dot a}}} \right\|}^2}} }}{2} \hfill \\
  {\lambda _2} = \frac{{{N_t}\left( {{{\left\| {{\mathbf{\dot b}}} \right\|}^2} + {{\left| \beta  \right|}^2}{N_r}} \right) + {N_r}{{\left\| {{\mathbf{\dot a}}} \right\|}^2} - \sqrt {{{\left( {{N_t}\left( {{{\left\| {{\mathbf{\dot b}}} \right\|}^2} + {{\left| \beta  \right|}^2}{N_r}} \right) - {N_r}{{\left\| {{\mathbf{\dot a}}} \right\|}^2}} \right)}^2} + 4{{\left| \beta  \right|}^2}{N_t}N_r^2{{\left\| {{\mathbf{\dot a}}} \right\|}^2}} }}{2} \hfill \\
\end{gathered}
\end{equation}

\setcounter{equation}{\value{MYtempeqncnt3}}
\hrulefill
\vspace*{4pt}
\end{figure*}
\\\indent Now let us take a close look at $\mathbf{Z}_k$. From (\ref{eq62}) we have
\setcounter{equation}{65}
\begin{equation}\label{eq67}
\begin{gathered}
  {{\mathbf{Z}}_k} = {\mu _{T}}{{\mathbf{I}}_{{N_t}}} - \left( {{\mathbf{F}} - \sum\nolimits_{i = 1}^K {{\mu _i}{\Gamma _i}{{\mathbf{Q}}_i}} } \right) - {\mu _k}\left( {1 + {\Gamma _k}} \right){{\mathbf{Q}}_k} \hfill \\
   \triangleq {\mu _{T}}{{\mathbf{I}}_{{N_t}}} - {\mathbf{\bar F}} - {\mu _k}\left( {1 + {\Gamma _k}} \right){{\mathbf{Q}}_k} \succeq \mathbf{0}, \forall k,\hfill \\
\end{gathered}
\end{equation}
where ${\mathbf{\bar F}} \triangleq {{\mathbf{F}} - \sum\nolimits_{i = 1}^K {{\mu _i}{\Gamma _i}{{\mathbf{Q}}_i}} }$. Note that $\mathbf{Z}_k \succeq \mathbf{0}$ implies that
\begin{equation}\label{eq68}
  {\mu _{T}} \ge {\lambda _{\max }}\left( {{\mathbf{\bar F}} + {\mu _k}\left( {1 + {\Gamma _k}} \right){{\mathbf{Q}}_k}} \right),\forall k,
\end{equation}
where $\lambda_{\max}\left(\cdot\right)$ represents the largest eigenvalue of the input matrix. By observing $\operatorname{tr}\left({\mathbf{Z}}_k{{\mathbf{W}}_k}\right) = 0$ from (\ref{eq58}), $\mathbf{Z}_k$ must be singular since $\mathbf{W}_k \succeq \mathbf{0}, \mathbf{W}_k \ne \mathbf{0}$. Therefore we have
\begin{equation}\label{eq69}
  {\mu _{T}} = {\lambda _{\max }}\left( {{\mathbf{\bar F}} + {\mu _k}\left( {1 + {\Gamma _k}} \right){{\mathbf{Q}}_k}} \right),\forall k.
\end{equation}
Apparently, the rank of ${\mathbf{Z}_k}$ strongly depends on the value of $\mu_k \ge 0$. Let us split the index set $\mathcal{K} = \left\{1,2,\ldots,K \right\}$ into two subsets, i.e.,
\begin{equation}\label{eq78-1}
  {\mathcal{K}_1} = \left\{ {\left. {{k}} \right|{\mu _k} > 0,\forall k} \right\},{\mathcal{K}_2} = \left\{ {\left. {{k}} \right|{\mu _k} = 0,\forall k} \right\}.
\end{equation}
We have $\mathcal{K} = \mathcal{K}_1 \bigcup \mathcal{K}_2$. Next, we discuss the rank of $\mathbf{Z}_k$ and $\mathbf{W}_k$ under the following cases, which cover all the possible values that $\mu_k,\forall k$ may take.
\\\indent {1) \textbf{Case I}: $\left| {{\mathcal{K}_1}} \right| = 0$.}
\\\indent In this case, we have $\mu_k = 0, \forall k$, and all the SINR constraints are ineffective. It follows that
\begin{equation}\label{eq70}
  {{\mathbf{Z}}_k} = {\mu _{T}}{{\mathbf{I}}_{{N_t}}} - {\mathbf{F}}\succeq \mathbf{0},\forall k.
\end{equation}
Again, $\mathbf{Z}_k$ must be singular to ensure a non-zero $\mathbf{W}_k$, which leads to
\begin{equation}\label{eq70-1}
  {\mu _{T}} = {\lambda _{\max }}\left( {\mathbf{F}} \right) = {\lambda _1}.
\end{equation}
Since $\lambda_1 > \lambda_2$, it holds immediately that $\operatorname{rank}\left(\mathbf{Z}_k\right) = N_t - 1, \operatorname{rank}\left(\mathbf{W}_k\right) = 1,\forall k$.
\\\indent {2) \textbf{Case II}: $\left| {{\mathcal{K}_1}} \right| = 1$.}
\\\indent In this case, only one $\mu_k$ is strictly positive, and the remaining ones are zero. Without loss of generality, let $\mu_1 >0, \mu_k = 0, \forall k \ge 2$ for notational convenience. We can express $\mathbf{Z}_k$ as
\begin{equation}\label{eq71-1}
\begin{gathered}
  {{\mathbf{Z}}_1} = {\mu _{T}}{{\mathbf{I}}_{{N_t}}} - {\mathbf{F}} - {\mu _1}{{\mathbf{Q}}_1} \succeq \mathbf{0}, \hfill \\
  {{\mathbf{Z}}_k} = {\mu _{T}}{{\mathbf{I}}_{{N_t}}} - {\mathbf{F}} + {\mu _1}{\Gamma _1}{{\mathbf{Q}}_1} \succeq \mathbf{0},\forall k > 1. \hfill \\
\end{gathered}
\end{equation}
It follows that
\begin{equation}\label{eq72-1}
  {\mu _{T}} = {\lambda _{\max }}\left( {{\mathbf{F}} + {\mu _1}{{\mathbf{Q}}_1}} \right) = {\lambda _{\max }}\left( {{\mathbf{F}} - {\mu _1}{\Gamma _1}{{\mathbf{Q}}_1}} \right).
\end{equation}
Given the semidefiniteness of ${{\mathbf{Q}}_1} = {{\mathbf{h}}_1}{\mathbf{h}}_1^H$, (\ref{eq72-1}) holds only if ${\mathbf{f}}_{\max }^H{{\mathbf{h}}_1} = 0$, where $\mathbf{f}_{\max}$ is the eigenvector of $\mathbf{F}$ corresponding to the largest eigenvalue $\lambda_1$. This also implies that
\begin{equation}\label{eq73-1}
  {\mu _{T}} = {\lambda _1} \Leftrightarrow \operatorname{rank} \left( {{\mu _{T}}{{\mathbf{I}}_{{N_t}}} - {\mathbf{F}}} \right) = {N_t} - 1.
\end{equation}
Hence, we have
\begin{equation}\label{eq74-1}
\begin{gathered}
  {N_t} - 2 \le \operatorname{rank} \left( {{{\mathbf{Z}}_1}} \right) \le {N_t} - 1,\hfill \\
  \operatorname{rank} \left( {{{\mathbf{Z}}_k}} \right) = {N_t} - 1,\forall k \ge 2. \hfill \\
\end{gathered}
\end{equation}
By recalling $\operatorname{tr} \left( {{{\mathbf{Z}}_k}{{\mathbf{W}}_k}} \right) = 0$, ${{\mathbf{W}}_k}$ should be within the null-space of $\mathbf{Z}_k$, which is
\begin{equation}\label{eq75-1}
  \begin{gathered}
  \mathcal{N}\left( {{{\mathbf{Z}}_1}} \right) = \operatorname{span} \left\{ {{{\mathbf{h}}_1},{{\mathbf{f}}_{\max }}} \right\}, \hfill \\
  \mathcal{N}\left( {{{\mathbf{Z}}_k}} \right) = \operatorname{span} \left\{ {{{\mathbf{f}}_{\max }}} \right\},\forall k \ge 2. \hfill \\
\end{gathered}
\end{equation}
Accordingly, $\mathbf{W}_k$ can be expressed as
\begin{equation}\label{eq76-1}
  {{\mathbf{W}}_1} = {a_1}{{\mathbf{h}}_1}{\mathbf{h}}_1^H + {b_1}{{\mathbf{f}}_{\max }}{\mathbf{f}}_{\max }^H,
  {{\mathbf{W}}_k} = {b_k}{{\mathbf{f}}_{\max }}{\mathbf{f}}_{\max }^H, \forall k \ge 2,
\end{equation}
where $a_k \ge 0, b_k \ge 0, \forall k$.
\\\indent If $a_1 = 0$, then $\operatorname{rank}\left(\mathbf{W}_k\right) = 1, \forall k$ holds. Otherwise if $a_1 > 0$, let
\begin{equation}\label{eq77-1}
\begin{gathered}
  {{{\mathbf{W'}}}_1} = {a_1}{{\mathbf{h}}_1}{\mathbf{h}}_1^H, {{{\mathbf{W'}}}_2} = \left( {{b_1} + {b_2}} \right){{\mathbf{f}}_{\max }}{\mathbf{f}}_{\max }^H,\hfill \\
  {{{\mathbf{W'}}}_k} = {{\mathbf{W}}_k},\forall k \ge 3. \hfill \\
\end{gathered}
\end{equation}
It can be readily verified that if $\left\{ {{{\mathbf{W}}_k}} \right\}_{k = 1}^K$ is an optimal solution for problem (\ref{eq34}), then $\left\{ {{{\mathbf{W'}}_k}} \right\}_{k = 1}^K$ is a rank-one optimal solution, due to the fact that the following conditions hold
\begin{subequations}\label{eq77-2}
 \begin{align}
  &\sum\nolimits_{k = 1}^K {{{{\mathbf{W'}}}_k}}  = \sum\nolimits_{k = 1}^K {{{\mathbf{W}}_k}}  = {{\mathbf{R}}_X},  \hfill \\
  &\begin{gathered}
  \left( {1 + {\Gamma _2}} \right)\operatorname{tr} \left( {{{\mathbf{Q}}_2}{{{\mathbf{W'}}}_2}} \right) - {\Gamma _2}\operatorname{tr} \left( {{{\mathbf{Q}}_2}{{\mathbf{R}}_X}} \right) \hfill \\
   \ge \left( {1 + {\Gamma _2}} \right)\operatorname{tr} \left( {{{\mathbf{Q}}_2}{{\mathbf{W}}_k}} \right) - {\Gamma _2}\operatorname{tr} \left( {{{\mathbf{Q}}_2}{{\mathbf{R}}_X}} \right) \ge {\Gamma _2}\sigma _C^2, \hfill \\
\end{gathered}  \hfill \\
  &\begin{gathered}
  \left( {1 + {\Gamma _k}} \right)\operatorname{tr} \left( {{{\mathbf{Q}}_k}{{{\mathbf{W'}}}_k}} \right) - {\Gamma _k}\operatorname{tr} \left( {{{\mathbf{Q}}_k}{{\mathbf{R}}_X}} \right) \hfill \\
   = \left( {1 + {\Gamma _k}} \right)\operatorname{tr} \left( {{{\mathbf{Q}}_k}{{\mathbf{W}}_k}} \right) - {\Gamma _1}\operatorname{tr} \left( {{{\mathbf{Q}}_k}{{\mathbf{R}}_X}} \right) \ge {\Gamma _k}\sigma _C^2,\forall k \ne 2. \hfill \\
\end{gathered}
\end{align}
\end{subequations}
Note that the objective value does not change due to (\ref{eq77-2}a). Moreover, (\ref{eq77-2}b) and (\ref{eq77-2}c) guarantee the feasibility of $\left\{ {{{\mathbf{W'}}_k}} \right\}_{k = 1}^K$. Therefore, in Case II, rank-one optimal solutions can always be attained.
\\\indent {3) \textbf{Case III}: $\left| {{\mathcal{K}_1}} \right| \ge 2$.}
\\\indent In this case, we observe that there are at least two positive $\mu_k$, and the remaining ones might be zero or positive. Let $M \triangleq \left| {{\mathcal{K}_1}} \right| \ge 2$, and assume, without loss of generality, that
\begin{equation}\label{eq78-1}
  {\mathcal{K}_1} = \left\{ 1,2,\ldots,M \right\},{\mathcal{K}_2} = \left\{M+1,M+2,\ldots,K \right\}.
\end{equation}
Let ${\mathbf{\tilde H}} = {\left[ {{{\mathbf{h}}_1},{{\mathbf{h}}_2}, \ldots, {{\mathbf{h}}_M}} \right]^H}$, ${\mathbf{\bar A}} = \left[ {{\mathbf{a}},{\mathbf{\dot a}}} \right]$, and define ${\mathbf{D}} \triangleq {\mathbf{\tilde H}}{\mathbf{\bar A}}$. The following lemma holds.
\begin{lemma}
  If ${\mathbf{D}}$ is of full column rank, then ${\mu _{T}}{{\mathbf{I}}_{{N_t}}} - {\mathbf{\bar F}}\succ \mathbf{0}$.
\end{lemma}
\begin{proof}
  We shall prove this lemma by contradiction. First we note that ${\mu _{T}}{{\mathbf{I}}_{{N_t}}} - {\mathbf{\bar F}}\succeq \mathbf{0}$ holds from (\ref{eq68}), given the non-negativity of ${\mu _k}\left( {1 + {\Gamma _k}} \right){{\mathbf{Q}}_k}$. It follows that
  \begin{equation}\label{eq79-1}
    {\mu _{T}} \ge {\lambda _{\max }}\left( {{\mathbf{\bar F}}} \right).
  \end{equation}
  Now suppose that ${\mu _{T}}{{\mathbf{I}}_{{N_t}}} - {\mathbf{\bar F}}\succ \mathbf{0}$ does not hold, and hence ${\mu _{T}} = {\lambda _{\max }}\left( {{\mathbf{\bar F}}} \right)$. Let ${\mathbf{\bar f}}_{\max}$ be the eigenvector of ${{\mathbf{\bar F}}}$ corresponding to ${\lambda _{\max }}\left( {{\mathbf{\bar F}}} \right)$, i.e.,
  \begin{equation}\label{eq80-1}
    {{\mathbf{\bar f}}_{\max} ^H}{\mathbf{\bar F}}{\mathbf{\bar f}}_{\max} = {\lambda _{\max }}\left( {{\mathbf{\bar F}}} \right) = \mu _{T}, {\left\| {{\mathbf{\bar f}}_{\max}} \right\|^2} = 1.
  \end{equation}
  We then have
  \begin{equation}\label{eq81-1}
  \begin{gathered}
  \mu _{T} + {\mu _k}\left( {1 + {\Gamma _k}} \right){{\mathbf{\bar f}}_{\max} ^H}{{\mathbf{Q}}_k}{\mathbf{\bar f}}_{\max}  \hfill \\ = {{\mathbf{\bar f}}_{\max} ^H}\left( {{\mathbf{\bar F}} + {\mu _k}\left( {1 + {\Gamma _k}} \right){{\mathbf{Q}}_k}} \right){\mathbf{\bar f}}_{\max} \hfill \\
   \le {\lambda _{\max }}\left( {{\mathbf{\bar F}} + {\mu _k}\left( {1 + {\Gamma _k}} \right){{\mathbf{Q}}_k}} \right) = \mu_{T}, \hfill \\
  \end{gathered}
  \end{equation}
  where the first equality holds from the assumption, and the last inequality in (\ref{eq81-1}) holds by the definition of the largest eigenvalue. Given $\mu_k > 0, \forall k \le M$, we have
  \begin{equation}\label{eq82-1}
    {{{\mathbf{\bar f}}_{\max}}^H}{{\mathbf{Q}}_k}{{\mathbf{\bar f}}_{\max}} = 0 \Leftrightarrow {\mathbf{h}}_k^H{{\mathbf{\bar f}}_{\max}} = 0,\forall k \le M.
  \end{equation}
  Due to the definitions of $\mathbf{F}$ and $\mathbf{\bar F}$, we know that
  \begin{equation}\label{eq83-1}
    {{\mathbf{\bar f}}_{\max}} \in \operatorname{span}\left\{ {{\mathbf{a}},{\mathbf{\dot a}},{{\mathbf{h}}_1},...,{{\mathbf{h}}_M}} \right\}.
  \end{equation}
  By both (\ref{eq82-1}) and (\ref{eq83-1}), we see that
  \begin{equation}\label{eq84-1}
    {{\mathbf{\bar f}}_{\max}} \in \operatorname{span}\left\{ {{\mathbf{a}},{\mathbf{\dot a}}} \right\}.
  \end{equation}
  Suppose that ${{\mathbf{\bar f}}_{\max}} =  {f _1}{\mathbf{a}} + {f _2}{\mathbf{\dot a}}$. From (\ref{eq82-1}) we have
  \begin{equation}\label{eq85-1}
    {\mathbf{h}}_k^H\left[ {{\mathbf{a}},{\mathbf{\dot a}}} \right]\left[ \begin{gathered}
  {f _1} \hfill \\
  {f _2} \hfill \\
\end{gathered}  \right] = 0,\forall k \le M \Leftrightarrow \mathbf{D}\left[ \begin{gathered}
  {f _1} \hfill \\
  {f _2} \hfill \\
\end{gathered}  \right] = {\mathbf{0}}.
  \end{equation}
  Hence, ${\mathbf{D}}$  is not of full column rank, which leads to contradiction. This completes the proof.
\end{proof}
\begin{corollary}
  If $\left| {{\mathcal{K}_1}} \right| \ge 2$, and ${\mathbf{D}}$ is of full column rank, then $\left| {{\mathcal{K}_1}} \right| = K$, and $\operatorname{rank}\left(\mathbf{W}_k\right) = 1, \forall k$.
\end{corollary}
\begin{proof}
  In light of Lemma 3, for those ${k} \in {\mathcal{K}_2}$, we have
  \begin{equation}\label{eq86-1}
    {{\mathbf{Z}}_k} = {\mu _{T}}{{\mathbf{I}}_{{N_t}}} - {\mathbf{\bar F}} \succ {\mathbf{0}}, \forall k \ge M+1.
  \end{equation}
  This suggests that $\mathbf{W}_k = \mathbf{0}, \forall k \ge M+1$, which is infeasible. Therefore, all $\mu_k$ should be positive to ensure that $\mathbf{Z}_k$ is singular, in which case $\left| {{\mathcal{K}_2}} \right| = 0$, and thereby $\left| {{\mathcal{K}_1}} \right| = K$. In this case, all $\mathbf{Z}_k$ can be expressed as in (\ref{eq67}), with $\mu_k >0, \forall k$. As a consequence, $\operatorname{rank}\left(\mathbf{Z}_k\right) = N_t-1, \forall k$, as each of them is obtained by subtracting a rank-one semidefinite matrix from a full-rank positive-definite matrix. This implies that $\operatorname{rank}\left(\mathbf{W}_k\right) = 1, \forall k$, which completes the proof.
\end{proof}
By Corollary 1, in Case III we have rank-one solutions if $\mathbf{D}$ is of full column rank. Given the above discussions on all the three cases, it holds true that if $\mathbf{H}\mathbf{\bar A}$ is of full column rank, solving problem (\ref{eq34}) always yields rank-one solutions, which completes the proof of Theorem 2.

\section{Proof of Theorem 3}
In light of Lemma 2, the optimization problem (\ref{eq38}) can be equivalently formulated as
\begin{equation}\label{eq71}
\begin{gathered}
  \mathop {\min }\limits_{\left\{ {{\lambda_{ii}}} \right\}_{i = 1}^{{N_t}}} \sum\nolimits_{i = 1}^{{N_t}} {\lambda_{ii}^{ - 1}}  \hfill \\
  \;\;\;{\text{s.t.}}\;\;\;\;{\lambda_{11}}{\left\| {{{\mathbf{h}}_{1}}} \right\|^2} \ge {\Gamma _1}\sigma _C^2, \hfill \\
  \;\;\;\;\;\;\;\;\sum\nolimits_{i = 1}^{{N_t}} {{\lambda_{ii}}}  \le {P_T}, \lambda_{ii} \ge 0, \forall i.\hfill \\
\end{gathered}
\end{equation}
Problem (\ref{eq71}) is convex. Note that we should have $\frac{{{\Gamma _1}\sigma _C^2}}{{{{\left\| {{{\mathbf{h}}_{\mathbf{1}}}} \right\|}^2}}} \le {\lambda _{11}} \le {P_T} \Leftrightarrow {\Gamma _1} \le \frac{{{P_T}{{\left\| {{{\mathbf{h}}_{\mathbf{1}}}} \right\|}^2}}}{{\sigma _C^2}}$ for ensuring the feasibility of the problem. In order to find a closed-form solution, we formulate the Lagrangian of the problem as follows
\begin{equation}\label{eq72}
\begin{gathered}
  \mathcal{L} = \sum\nolimits_{i = 1}^{{N_t}} {\lambda_{ii}^{ - 1}}  + \omega \left( { - {\lambda_{11}} + \frac{{{\Gamma _1}\sigma _C^2}}{{{{\left\| {{{\mathbf{h}}_{1}}} \right\|}^2}}}} \right) \hfill \\
  \;\;\;\;\;\;\;\;\; + \mu \left( {\sum\nolimits_{i = 1}^{{N_t}} {{\lambda_{ii}}}  - {P_T}} \right) - {\eta_i}{\lambda_{ii}}, \hfill \\
\end{gathered}
\end{equation}
where $\omega$, $\mu$, and $\eta_i,\forall i$ are the dual variables. Accordingly, the Karush-Kuhn-Tucker (KKT) conditions can be given as
\begin{subequations}\label{eq73}
 \begin{align}
  \frac{{\partial \mathcal{L}}}{{\partial {\lambda_{11}}}} =  - \lambda_{11}^{ - 2} - \omega  + \mu  - {\eta_1} = 0, & \hfill \\
  \frac{{\partial \mathcal{L}}}{{\partial {\lambda_{ii}}}} =  - \lambda_{ii}^{ - 2} + \mu  - {\eta_i} = 0,\;\;i = 2,3, \ldots {N_t}, &\hfill \\
  \omega \left( { - {\lambda_{11}} + \frac{{{\Gamma _1}\sigma _C^2}}{{{{\left\| {{{\mathbf{h}}_1}} \right\|}^2}}}} \right) = 0, \omega  \ge 0,{\lambda_{11}} \ge \frac{{{\Gamma _1}\sigma _C^2}}{{{{\left\| {{{\mathbf{h}}_1}} \right\|}^2}}}, &\hfill \\
  \mu \left( {\sum\nolimits_{i = 1}^{{N_t}} {{\lambda_{ii}}}  - {P_T}} \right) = 0,\mu  \ge 0,\sum\nolimits_{i = 1}^{{N_t}} {{\lambda_{ii}}}  \le {P_T}, &\hfill \\
  {\eta_i}{\lambda_{ii}} = 0,{\eta_i} \ge 0,{\lambda_{ii}} \ge 0. &
\end{align}
\end{subequations}
It can be immediately observed that $\eta_i = 0, \forall i$, since ${\lambda_{ii}} > 0, \forall i$. Therefore, (\ref{eq73}a) and (\ref{eq73}b) can be rewritten as
\begin{subequations}\label{eq74}
  \begin{align}
    \lambda_{11}^{ - 2} = \mu  - \omega , & \hfill \\
    \lambda_{ii}^{ - 2} = \mu,\;\;i = 2,3, \ldots, {N_t}. &
  \end{align}
\end{subequations}
From (\ref{eq74}b) we see that $\mu > 0$, which indicates that the power budget should always be reached, i.e., ${\sum\nolimits_{i = 1}^{{N_t}} {{\lambda_{ii}}}  = {P_T}}$. We then investigate the value of $\omega$. Suppose that $\omega = 0$, which suggests ${\lambda_{11}} > \frac{{{\Gamma _1}\sigma _C^2}}{{{{\left\| {{{\mathbf{h}}_1}} \right\|}^2}}}$, i.e., the SINR constraint is inactive. In this case, we have $\lambda_{ii}^{ - 2} = \mu, \forall i$. Hence, all $\lambda_{ii}$ should be the same, which is
\begin{equation}\label{eq75}
  \lambda_{ii} = \frac{P_T}{N_t}, \forall i.
\end{equation}
By using (\ref{eq75}) and Lemma 2 we obtain the optimal solution as
\begin{equation}\label{eq76}
  {{\mathbf{W}}_1} = \frac{{{P_T}}}{{{N_t}}}\frac{{{\mathbf{h}}_1{\mathbf{h}}_1^H}}{{{{\left\| {{{\mathbf{h}}_1}} \right\|}^2}}},{{\mathbf{R}}_X} = \frac{{{P_T}}}{{{N_t}}}{{\mathbf{I}}_{{N_t}}}.
\end{equation}
Note that this requires that the following condition holds
\begin{equation}\label{eq77}
  \frac{{{P_T}}}{{{N_t}}} > \frac{{{\Gamma _1}\sigma _C^2}}{{{{\left\| {{{\mathbf{h}}_{\mathbf{1}}}} \right\|}^2}}} \Leftrightarrow {\Gamma _1} < \frac{{{P_T}{{\left\| {{{\mathbf{h}}_{\mathbf{1}}}} \right\|}^2}}}{{{N_t}\sigma _C^2}}.
\end{equation}
On the other hand, if $\omega > 0$, the SINR constraint is active and we have ${\lambda_{11}} = \frac{{{\Gamma _1}\sigma _C^2}}{{{{\left\| {{{\mathbf{h}}_{\mathbf{1}}}} \right\|}^2}}}$. To satisfy the power constraint, it holds that $\sum\nolimits_{ii = 2}^{{N_t}} {{\lambda _{ii}}}  = {P_T} - \frac{{{\Gamma _1}\sigma _C^2}}{{{{\left\| {{{\mathbf{h}}_1}} \right\|}^2}}}$. By noting (\ref{eq74}b), we arrive at
\begin{equation}\label{eq78}
  {\lambda _{ii}} = \frac{{{P_T}{{\left\| {{{\mathbf{h}}_1}} \right\|}^2} - {\Gamma _1}\sigma _C^2}}{{{{\left\| {{{\mathbf{h}}_1}} \right\|}^2}\left( {{N_t} - 1} \right)}},i = 2,3, \ldots ,{N_t}.
\end{equation}
Hence, for the case that ${\Gamma _1} > \frac{{{P_T}{{\left\| {{{\mathbf{h}}_{\mathbf{1}}}} \right\|}^2}}}{{{N_t}\sigma _C^2}}$, we have the optimal solution in the form of
\begin{equation}\label{eq79}
 \begin{gathered}
  {{\mathbf{W}}_1} =  \left( {{\mathbf{u}}_1^H{{\mathbf{R}}_X}{{\mathbf{u}}_1}} \right){{\mathbf{u}}_1}{\mathbf{u}}_1^H = \frac{{{\Gamma _1}\sigma _C^2{\mathbf{h}}_1{\mathbf{h}}_1^H}}{{{{\left\| {{{\mathbf{h}}_1}} \right\|}^4}}}, \hfill \\
  {{\mathbf{R}}_X} = \sum\nolimits_{i = 1}^{{N_t}} {{\lambda _{ii}}{{\mathbf{u}}_i}} {\mathbf{u}}_i^H, \hfill \\
\end{gathered}
\end{equation}
which completes the proof.
\ifCLASSOPTIONcaptionsoff
  \newpage
\fi



\bibliographystyle{IEEEtran}
\bibliography{IEEEabrv,JRC_REF}
\end{document}